\documentclass[a4paper,12pt]{article}
\usepackage{amssymb}
\usepackage{amsthm}
\usepackage{amsmath}
\usepackage[utf8]{inputenc}
\usepackage{array}
\usepackage[T1]{fontenc}
\usepackage{lmodern}
\usepackage{tikz}
\usepackage{vmargin}
\usepackage{enumerate}
\usepackage{enumitem}
\setmarginsrb{2.5cm}{1cm}{2.5cm}{1cm}{1cm}{1cm}{1cm}{1cm}

\newtheorem{theo}{Theorem}
\newtheorem*{theo*}{Theorem}
\newtheorem*{prop*}{Property}
\newtheorem{cor}[theo]{Corollary}
\newtheorem{lemm}[theo]{Lemma}

\newtheorem{proc}[theo]{Procedure}

\newcommand{\cC}{{\cal C}} %% 
\newcommand{\cI}{{\cal I}} %% 
\newcommand{\cO}[1]{{\cal O}_{#1}} %% 
\newcommand{\cP}[1]{{\cal P}_{#1}} %% 
\newcommand{\IO}[1]{({\cI, \cO{#1}})} %% (I, O_k)

\DeclareMathOperator{\mad}{mad}

\begin{document}

\title{Partitioning sparse graphs into an independent set and a graph with bounded size components}

\author{
Ilkyoo Choi\thanks{
Department of Mathematics, Hankuk University of Foreign Studies, Yongin-si,
Gyeonggi-do, Republic of Korea, \texttt{ilkyoo@hufs.ac.kr}.
Supported by the Basic Science Research Program through
the National Research Foundation of Korea (NRF) funded by the Ministry of
Education (NRF-2018R1D1A1B07043049), and also by Hankuk University of Foreign
Studies Research Fund.
}
\and
Fran\c{c}ois Dross\thanks{
Université Côte d'Azur, I3S, CNRS, Inria, France
\texttt{Francois.Dross@googlemail.com}.
}
\and
Pascal Ochem\thanks{
LIRMM, CNRS, Montpellier, France,
\texttt{pascal.ochem@lirmm.fr}.
The last two authors were partially supported by
the ANR grant HOSIGRA (contract number ANR-17-CE40-0022-03).
}
}

\date\today

\maketitle

\begin{abstract}
We study the problem of partitioning the vertex set of a given graph so that each part induces a graph with components of bounded order; we are also interested in restricting these components to be paths. 
In particular, we say a graph $G$ admits an $\IO{k}$-partition if its vertex set can be partitioned into an independent set and a set that induces a graph with components of order at most $k$.
We prove that every graph $G$ with $\mad(G)<\frac 52$ admits an $\IO{3}$-partition.
This implies that every planar graph with girth at least $10$ can be partitioned into an independent set and a set that induces a graph whose components are paths of order at most 3. 
We also prove that every graph $G$ with $\mad(G) < \frac{8k}{3k+1} = \frac{8}{3}\left( 1 - \frac{1}{3k+1} \right)$ admits an $\IO{k}$-partition. 
This implies that every planar graph with girth at least $9$ can be partitioned into an independent set and a set that induces a graph whose components are paths of order at most 9. 
\end{abstract}

Given a graph $G$, let $V(G)$ and $E(G)$ be the vertex set and the edge set, respectively, of $G$. 
%A \emph{triangle} is a cycle of length 3. 
%The \emph{average degree} of a graph $G$ is $\ad(G) = \frac{|E(G)|}{|V(G)|}$. 
The \emph{maximum average degree} of a graph $G$, denoted $\mad(G)$, is the maximum of $\frac{2|E(H)|}{|V(H)|}$ over all subgraphs $H$ of $G$.
We are interested in partitioning the vertex set of a sparse graph $G$, measured in terms of $\mad(G)$, so that each part induces a graph with bounded component sizes. 
There is also a substantial line of research investigating partitions of the planar graphs into parts that induce graphs with bounded maximum degree; we refer the interested readers to~\cite{2014BoKo,2017ChChJeSu,2015ChRa,2000Skrekovski}.

We will use the following notation to denote the following classes of graphs: 
\begin{align*}
\cI_{\;\,\,}&: {\mbox{the class of edgeless graphs}}\\
\cO{k}&: {\mbox{the class of graphs whose components have order at most $k$}}\\
\cP{k}&: {\mbox{the class of graphs whose components are paths of order at most $k$}}
\end{align*}
Note that $\cI\subset\cP{k}\subset\cO{k}$, and $\cI$, $\cO{1}$, and $\cP{1}$ all denote the class of edgeless graphs. 
In the above notation, the classical chromatic number corresponds to partitioning the vertex set of a graph so that each part induces a graph in $\cI$.
For graph classes $\cC_1, \ldots, \cC_k$, we say a graph $G$ admits a \emph{$(\cC_1, \ldots, \cC_k)$-partition} if $V(G)$ can be partitioned into $k$ sets $V_1, \ldots, V_k$ so that $G[V_i]$ is a graph in $\cC_i$ for all $i$. 

The celebrated Four Colour Theorem~\cite{1977ApHa,1977ApHaKo} states that every planar graph admits an $(\cI, \cI, \cI, \cI)$-partition.
On the other hand, Alon et al.~\cite{2003AlDiOpVe} proved that there is no finite $k$ where every planar graphs admit an $(\cO{k}, \cO{k}, \cO{k})$-partition. 
Gr\"otzsch's Theorem~\cite{1959Grotzsch} shows that every planar graph with girth 4 admits an $(\cI, \cI, \cI)$-partition, yet, for a given $k$, it is not hard to construct a graph with girth 4 that does not admit an $(\cO{k}, \cO{k})$-partition; for example, see~\cite{2015MoOc}.
A recent result by Dvo\v r\'ak and Norin~\cite{DvNo} implies that every planar graph with girth 5 admits an $(\cO{k}, \cO{k})$-partition for some finite $k$; this was conjectured by Esperet and Ochem in~\cite{esperet2016islands}.

%\begin{theo}[Esperet and Ochem~\cite{esperet2016islands}] \label{TH_EO}
%Every planar graph with girth at least $6$ admits an $({\cal O}_{12},{\cal O}_{12})$-partition.
%\end{theo}
%
%%%%%%%%%% theorem 8 in Esperet+Ochem says component size 16?

We direct the attention to partitioning the vertex set of a planar graph so that each part induces a graph whose components are paths of bounded order. 
The Four Colour Theorem~\cite{1977ApHa,1977ApHaKo} again tells us that we can partition a planar graph into four parts where each part induces a graph whose components are paths. 
It is known~\cite{1991Goddard,1990Poh} that the vertex set of a planar graph can be partitioned into three parts where each part induces a graph whose components are paths, yet the order of the paths cannot be bounded. 

Regarding planar graphs with girth restrictions, Borodin, Kostochka, and Yancey~\cite{2013BoKoYa} proved that a planar graph with girth at least 7 admits a $(\cP{2}, \cP{2})$-partition, improving upon a result by Borodin and Ivanova~\cite{borodin2011list} stating that it admits a $(\cP{3}, \cP{3})$-partition. 
Recently, Axenovich, Ueckerdt, and Weiner~\cite{2017AxUeWe} showed that a planar graph with girth at least 6 admits a $(\cP{15}, \cP{15})$-partition, whereas it was conjectured in~\cite{borodin2011list} that it admits a $(\cP{3}, \cP{3})$-partition. 
As far as we know, it is currently unknown if there is a finite $k$ where a planar graph with girth at least 5 admit a $(\cP{k}, \cP{k})$-partition. 

%In \cite{borodin2011list}, it was shown that every planar graph with girth at least $7$ admits an $(\cO{3},\cO{3})$-partition.
%Since planar graphs with girth at least $7$ contain no 3-cycle, we can rephrase the aforementioned result as the following: 
%
%\begin{theo}[\cite{borodin2011list}]
%Every planar graph with girth at least $7$ admits a $(\cP{3},\cP{3})$-partition.
%\end{theo}

\bigskip

We are interested in the theme mentioned in the above paragraphs. 
However, we concentrate on the situation where one part must be an independent set. 
We obtain two results in terms of the maximum average degree. 
Since every planar graph with girth at least $g$ has maximum average degree less than $\frac{2g}{g-2}$ by Euler's formula, our results imply corollaries regarding planar graphs with certain girth restrictions. 

Our first result concerns $(\cI,\cO{3})$-partitions.

\begin{theo}\label{main}
Every graph $G$ with $\mad(G) < \frac{5}{2}$ admits an $(\cI,\cO{3})$-partition.
\end{theo}

%\begin{cor} \label{c_main}
%Every planar graph with girth at least $10$ admits an $(\cI,\cO{3})$-partition.
%\end{cor}

Since planar graphs with girth at least $10$ do not have 3-cycles, we obtain the following corollary:
%\begingroup
%\def\thetheo{\ref{c_main}}
\begin{cor}\label{c_main}%[restated]
Every planar graph with girth at least $10$ admits an $(\cI,\cP{3})$-partition.
\end{cor}
%\addtocounter{theo}{-1}
%\endgroup

Our next result is about $(\cI, \cO{k})$-partitions.  

\begin{theo} \label{t32_main}
Let $k \ge 2$ be an integer. Every graph $G$ with $\mad(G) < \frac{8k}{3k+1} = \frac{8}{3}\left( 1 - \frac{1}{3k+1} \right)$ admits an $(\cI,\cO{k})$-partition.
\end{theo}

Note that as $k$ goes to infinity, the upper bound on the maximum average degree gets arbitrarily close to $\frac{8}{3}$, which corresponds to planar graphs of girth at least $8$. 
As above, we have the following corollary:

\begin{cor} \label{c32_main}
Every planar graph with girth at least $9$ admits an $(\cI,\cO{9})$-partition.
\end{cor}

We are also interested in the complexity aspect of the corresponding decision problems. 
We show that some decision problems are either trivial (the answer is always yes) or NP-complete.
\begin{theo}\label{kg}
Let $k\ge2$ and $g\ge3$ be fixed integers.
Either every planar graph with girth at least $g$ has an $(\cI,\cP{k})$-partition
or it is NP-complete to determine whether a planar graph with girth at least $g$ has an $(\cI,\cP{k})$-partition.
Either every planar graph with girth at least $g$ has an $(\cI,\cO{k})$-partition
or it is NP-complete to determine whether a planar graph with girth at least $g$ has an $(\cI,\cO{k})$-partition.
\end{theo}

\begin{theo}\label{3gd}
Let $d$ and $g\ge3$ be fixed integers.
Either every planar graph with girth at least $g$ and maximum degree at most $d$ has an $(\cI,\cP{3})$-partition
or it is NP-complete to determine whether a planar graph with girth at least $g$ and maximum degree at most $d$ has an $(\cI,\cP{3})$-partition.
\end{theo}

Theorem~\ref{3gd}, together with some graphs constructed in Section~\ref{complex}, yields the following corollary.

\begin{cor} \label{c-3gd}
Deciding whether a planar graph with girth at least $5$ and maximum degree at most $3$ has an $(\cI,\cP{3})$-partition is NP-complete.
Deciding whether a planar graph with girth at least $7$ and maximum degree at most $4$ has an $(\cI,\cP{3})$-partition is NP-complete.
\end{cor}

We will prove Theorem~\ref{main} and Theorem~\ref{t32_main} in Section~\ref{s_main} and Section~\ref{s_t32_main}, respectively.
Theorems~\ref{kg} and~\ref{3gd} will be proven in Section~\ref{complex}. 
We end this section with some notation and definitions. 

Given an {$(\cI,\cO{k})$-partition} of $G$, we will assume that $V(G)$ is partitioned into two parts $I$ and $O$ where $I$ is an independent set and $O$ induces a graph whose components have order at most $k$; we also call the sets $I$ and $O$ \emph{colours}, and a vertex in $I$ and $O$ is said to have colour $I$ and $O$, respectively. 
A \emph{$k$-vertex}, a \emph{$k^+$-vertex}, and a \emph{$k^-$-vertex} are a vertex of degree $k$, at least $k$, and at most $k$,  respectively.
Also, for a vertex $v$, a \emph{$k$-neighbour}, a \emph{$k^+$-neighbour}, and a \emph{$k^-$-neighbour} of $v$ are a neighbour of $v$ with degree $k$, at least $k$, and at most $k$, respectively.

For $S \subset V(G)$, let $G - S$ be the graph obtained from $G$ by removing the vertices of $S$ and all the edges incident to a vertex of $S$. 
If $S=\{x\}$, then let $G - x$ denote $G - \{x\}$. 
For a set $S$ of vertices satisfying $S \cap V(G) = \emptyset$, let $G + S$ be the graph obtained from $G$ by adding the vertices of $S$. 
If $S=\{x\}$, then let $G + x$ denote $G + \{x\}$.
For a set $E'$ of nonedges of $G$, let $G + E'$ be the graph constructed from $G$ by adding the edges of $E'$. 
If $E'=\{e\}$, then we denote $G + \{e\}$ by $G + e$. 
For $W \subset V(G)$, we denote by $G[W]$ the subgraph of $G$ induced by $W$. 
For a vertex $v$, the open neighbourhood of $v$, denoted $N(v)$, is the set of neighbours of $v$; the closed neighbourhood of $v$ is $N(v)\cup\{v\}$, which is denoted by $N[v]$.
%For all $d \in \mathbb{N}$, a \emph{$d$-vertex} is a vertex of degree $d$ in $G$, and for all vertex $v$, a \emph{$d$-neighbour} of $v$ is a vertex of degree $d$ that is a neighbour of $v$.
For a set $X$ of vertices, let us denote by $N(X) = \{y \in V \setminus X, \exists x \in X, xy \in E\}$ the neighbourhood of $X$ in $G$. 

\section{Proof of Theorem~\ref{main}} \label{s_main}

We will prove Theorem~\ref{main} via the discharging method. 
Let $G$ be a counterexample to Theorem~\ref{main} with as few vertices as possible.
The graph $G$ is connected, since otherwise at least one of its components would be a counterexample with fewer vertices than $G$. 
This further implies that every vertex of $G$ has degree at least 1. 

Given an (partial) $(\cI, \cO{3})$-partition of $G$, we say a vertex coloured $O$ is \emph{saturated} if it has either two neighbours in $O$ or a neighbour in $O$ with two neighbours in $O$. 
%We say a vertex coloured $O$ is \emph{partially saturated} if it has a neighbour in $O$ and is not saturated. 
We say a vertex coloured $O$ is \emph{unsaturated} if it has no neighbour coloured $O$.

\begin{lemm} \label{l_dge2}
Every vertex in $G$ has degree at least $2$.
\end{lemm}
\begin{proof}
%Let $v$ be a vertex of degree at most $1$ in $G$. Since $G$ is connected, if $v$ has degree $0$, then $V(G) = \{v\}$ and $G$ admits an $({\cal I},{\cal O}_3)$-partition, a contradiction. Therefore $v$ has degree $1$. 
Let $v$ be a vertex of degree 1 in $G$.
Since the graph $G - v$ has fewer vertices than $G$, it admits an $(\cI,\cO{3})$-partition, which can be extended to an $(\cI,\cO{3})$-partition of $G$ by giving to $v$ the colour distinct from that of its neighbour. 
This contradicts that $G$ is a counterexample. 
\end{proof}

\begin{lemm} \label{l_3-}
Every $3^-$-vertex in $G$ has at least one $3^+$-neighbour.
\end{lemm}

\begin{proof}
Let $v$ be a $3^-$-vertex where every neighbour has degree $2$ and let $G' = G - N[v]$. 
Since a neighbour of $v$ is a $2$-vertex, it cannot have neighbours in both $N(v)$ and $G'$.
Therefore, the neighbours of $v$ that have a neighbour in $G'$ must form an independent set. 
Since $G'$ has fewer vertices than $G$, it admits an $(\cI,\cO{3})$-partition. 
For every neighbour $u$ of $v$ that has a neighbour $u'$ in $G'$, colour $u$ with the colour distinct from that of $u'$. 

First, colour all uncoloured veritces with $O$.
If this does not give an $(\cI,\cO{3})$-partition of $G$, it must be that all four vertices in $N[v]$ received the colour $O$, so we can recolour $v$ with $I$ to obtain an $(\cI,\cO{3})$-partition of $G$, which is a contradiction. 
%If two vertices of $N(v)$ are adjacent, we can colour one with colour $O$ and the other with colour $I$. If all the vertices in $N(v)$ have colour $O$, then we can colour $v$ with colour $I$. Otherwise, $v$ has at most two neighbours in $G$ that are coloured $O$, and each of them has no neighbours in $G$ that are coloured $O$.
%Therefore we can assign $O$ to $v$ and we obtain an $({\cal I},{\cal O}_3)$-partition of $G$, a contradiction.
\end{proof}

In $G$, a \emph{chain} is a longest induced path whose internal vertices all have degree $2$.
% in $G$ is called a \emph{chain}. 
A chain with $k$ internal vertices is a \emph{$k$-chain}. 
The \emph{end-vertices} of a $k$-chain are the vertices that are not internal vertices; 
note that every end-vertex of a chain is a $3^+$-vertex. 
By Lemma~\ref{l_3-}, there are no $3$-chains in $G$.

We will build a subgraph $L$ of $G$ using the following rules:

\begin{enumerate}[label={\bf Step \arabic*}]
\item For every $2$-chain $u_1v_1v_2u_2$, add all vertices to $L$, as well as the edges $u_1v_1$ and $u_2v_2$. 
Note that we do not add the edge $v_1v_2$. 
We set that, for $i \in \{1,2\}$, $u_i$ is the \emph{father} of $v_i$, and $v_i$ is a \emph{leaf} of $L$.

%\item Whenever there is a $4$-vertex $w$ that is the father of three vertices in $L$, such that the fourth neighbour of $w$ is a $2$-vertex $v$ that is not in $L$, do the following. Let $u$ be the neighbour of $v$ distinct from $w$. Since $v$ is not in $L$, $v$ is not in a $2$-chain, and thus $u$ is a $3^+$-vertex. Add the vertices $u$ (if it is not already in $L$) and $v$ to $L$, as well as the edges $uv$ and $vw$. The vertices $u$ and $v$ are the fathers of $v$ and $w$ respectively.
%
%Whenever there is a $3$-vertex $w$ that is the father of a vertex in $L$ and has a $2$-neighbour $v$ that is not in $L$, do the following. Let $u$ be the neighbour of $v$ distinct from $w$. Since $v$ is not in $L$, $v$ is not in a $2$-chain, and thus $u$ is a $3^+$-vertex. Add the vertices $u$ and $v$ to $L$, as well as the edges $uv$ and $vw$. The vertices $u$ and $v$ are the fathers of $v$ and $w$ respectively.
\item For every vertex $w$ that is either a 4-vertex and a father of three vertices in $L$ or a 3-vertex and a father of a vertex in $L$, if $w$ has a $2$-neighbour $v$  that is not in $L$, do the following:
 Let $u$ be the neighbour of $v$ distinct from $w$. Since $v$ is not in $L$, $v$ is not in a $2$-chain, and thus $u$ is a $3^+$-vertex. Add the vertices $u$ (if it is not already in $L$) and $v$ to $L$, as well as the edges $uv$ and $vw$. The vertices $u$ and $v$ are the fathers of $v$ and $w$, respectively.
 
 \item Repeat Step 3 as long as $L$ grows. 
  \end{enumerate}

A vertex $u$ is a \emph{son} of a vertex $v$ if $v$ is the father of $u$.
By construction of $L$, every son of a $3^+$-vertex is a $2$-vertex, and every son of a $2$-vertex is a $3^+$-vertex.

Since every $3$-vertex has a $3^+$-neighbour by Lemma~\ref{l_3-}, every $3$-vertex that has a son has at most one father.
Note that both a son and a father of a 3-vertex are $2$-vertices. 
This further implies that every vertex has at most one father. 
Hence, the construction above produces a rooted forest. 
We say that a vertex $v$ is a \emph{descendant} of a vertex $u$ if there are vertices $w_1 = u, w_2, \ldots, w_{k-1}, w_k = v$ for some $k\ge 2$, such that $w_i$ is the father of $w_{i+1}$ for each $i\in\{1, \ldots, k-1\}$.

\begin{lemm} \label{l_recol}
Given a $3^+$-vertex $v$ of $L$, let $u$ be the father of $v$ and let $D$ be the set of descendants of $v$. 
Let $H$ be an induced subgraph of $G - (D \cup \{u, v\})$.
We can find an $({\cal I},{\cal O}_3)$-partition of $G[V(H) \cup D \cup \{v\}]$ such that, if $v$ is coloured $O$, then it is not saturated.
\end{lemm}

\begin{proof}
Since $H$ has fewer vertices than $G$, the graph $H$ admits an $({\cal I},{\cal O}_3)$-partition. 

We will prove the lemma by induction. We assume that, for every $3^+$-vertex $x$ that is a descendant of $v$, we can colour $x$ and all of its descendants such that, if $x$ is coloured $O$, then it is not saturated.

Let $v'$ be a son of $v$. Recall $v'$ is a 2-vertex by construction of $L$. 
Let $w'$ be the neighbour of $v'$ distinct from $v$. 
If $v'$ is a leaf of $L$, then by construction of $L$, $w'$ is a 2-vertex. 
Let $w$ be the neighbour of $w'$ that is distinct from $v'$. 
If $w$ and $w'$ are both coloured $O$, then we can recolour $w'$ with $I$. 
Therefore, we can assume that if $w'$ is coloured $O$, then it is not saturated. 
We reach the same conclusion by induction if $v'$ is not a leaf of $L$.

Note that by the construction of $L$, since $v$ is a $3^+$-vertex that has a father, $v$ is either a $4$-vertex with three sons or a $3$-vertex with a son.
First, suppose that $v$ is a $4$-vertex with three sons $v_1$, $v_2$, and $v_3$.
For $i\in\{1,2,3\}$, let $w_i$ be the neighbour of $v_i$ distinct from~$v$. 
By above, we may assume for each $w_i$, if it is coloured $O$, then it is not saturated. 
Therefore, we can colour each $v_i$ with colour $O$ and colour $v$ with $I$. 
%Let $i \in \{1,2,3\}$. 
%Assume first that $v_i$ is a leaf of $L$. 
%By construction of $L$, $w_i$ is a $2$-vertex. Let $x$ be the neighbour of $w_i$ distinct from $v_i$. If $w_i$ and $x$ are both coloured $O$, then we can recolour $w_i$ to $I$. Therefore we can assume that, if $w_i$ is coloured $O$, then it is unsaturated. Assume now that $v_i$ is not a leaf of $L$. By induction, we can colour $w_i$ and its descendants such that, if $w_i$ is coloured $O$, then it is not saturated. Now, for all $i$, we can colour $v_i$ with colour $O$. Finally, we can colour $v$ with colour $I$.

Suppose now that $v$ is a $3$-vertex with a son $v_1$ and let $v_2$ be the third neighbour of $v$.
Let $w$ be the neighbour of $v_1$ distinct from $v$. 
By above, we know if $w$ is coloured $O$, then it is not saturated.
%As above, if $v_1$ is not a leaf, then we colour $w$ and its descendants such that if $w$ has colour $O$, then it is not saturated. 
%If $v_1$ is a leaf, then, as above, we can also assume that, if $w$ has colour $O$, then it is not saturated (otherwise we can recolour $w$ with $I$).
If $v_2$ is coloured $O$, then we can colour $v_1$ with colour $O$ and $v$ with colour $I$.
If $v_2$ is coloured $I$, then we can colour $v_1$ with a colour distinct from that of $w$, and colour $v$ with colour $O$. 
Note that in this case $v$ is not saturated. 

This proves the lemma.
\end{proof}

\begin{lemm} \label{l_3v}
Every $3$-vertex of $G$ has at most one son in $L$.
\end{lemm}

\begin{proof}
Let $v$ be a $3$-vertex with at least two sons, $v_1$ and $v_2$, in $L$. Let $w_1$ and $w_2$ be the neighbours of $v_1$ and $v_2$, respectively, distinct from $v$, and let $v_3$ be the third neighbour of $v$.
Note that by Lemma~\ref{l_3-}, $v_3$ is a $3^+$-vertex. Let $D$ be the set of the descendants of $v$. The graph $G - (D \cup\{v\})$ has fewer vertices than $G$. Therefore it admits an $({\cal I},{\cal O}_3)$-partition. For $i \in \{1,2\}$, if $v_i$ is a leaf in $L$, then we can assume, up to recolouring $w_i$, that if $w_i$ has colour $O$, then it is unsaturated. For $i \in \{1,2\}$, if $v_i$ is a not a leaf in $L$, then $w_i$ is a $3^+$-vertex, and by Lemma~\ref{l_recol}, we can colour $w_i$ and all of its descendants such that, if $w_i$ is coloured $O$, then it is not saturated. 

If $v_3$ has colour $I$, then, for all $i \in \{1,2\}$, we assign to $v_i$ a colour distinct from that of $w_i$, and since $v$ has at most two neighbours that have colour $O$ and each of them is unsaturated, we can assign colour $O$ to $v$.
If $v_3$ has colour $O$, then we assign colour $O$ to $v_1$ and $v_2$, and colour $I$ to $v$.
In both cases, that leads to an $({\cal I},{\cal O}_3)$-partition of $G$, a contradiction.
\end{proof}

\begin{lemm} \label{l_4v}
Every $4$-vertex of $G$ has at most three sons in $L$.
\end{lemm}

\begin{proof}
Let $v$ be a $4$-vertex in $L$ with four sons $v_1$, $v_2$, $v_3$, and $v_4$.
For $i\in\{1, 2, 3, 4\}$, let $w_i$ be the neighbour of $v_i$ that is distinct from $v$.
Let $D$ be the set of the descendants of $v$. 
The graph $G - (D \cup\{v\})$ has fewer vertices than $G$, so it admits an $({\cal I},{\cal O}_3)$-partition. 
For $i \in \{1,2,3,4\}$, if $v_i$ is a leaf in $L$, then we can assume that, up to recolouring $w_i$, if $w_i$ has colour $O$, then it is unsaturated. For $i \in \{1,2,3,4\}$, if $v_i$ is a not a leaf in $L$, then $w_i$ is a $3^+$-vertex, and by Lemma~\ref{l_recol}, we can colour $w_i$ and all of its descendants such that, if $w_i$ is coloured $O$, then it is not saturated. 
We can assign colour $O$ to each of the $v_i$'s, and colour $I$ to $v$.
This leads to an $({\cal I},{\cal O}_3)$-partition of $G$, a contradiction.
\end{proof}

\section*{Discharging procedure:}
We have all of our structural lemmas, so we can move on to the discharging procedure.
We assign an initial charge of $d-\frac{5}{2}$ to every $d$-vertex. Since $\mad(G)<\frac{5}{2}$, the sum of the initial charges in the graph is negative.

We now apply the following discharging procedure. 

\textbf{ \begin{enumerate}[label= Rule \arabic*]
\item {\it \emph{Every $3^+$-vertex of $L$ gives charge $\frac{1}{2}$ to each son.}}\label{inL}
\item {\it \emph{Every $3^+$-vertex gives charge $\frac{1}{4}$ to each $2$-neighbour that is not a neighbour in $L$.}}\label{ninL}
\end{enumerate}}

Throughout the discharging procedure, no charge is created and no charge is deleted. 
Since the total charge sum is preserved, the sum of the final charges of the vertices is negative. Let us now prove that every vertex has non-negative final charge. 
This contradiction completes the proof of Theorem~\ref{main}.

\begin{lemm}
Every vertex has a non-negative final charge.
\end{lemm}

\begin{proof}
Each $2$-vertex has initial charge $-\frac{1}{2}$. 
Every $2$-vertex in $L$ receives $\frac{1}{2}$ from its father by \ref{inL}; recall that every $2$-vertex in $L$ has a father in $L$ that is a $3^+$-vertex. 
Every $2$-vertex that is not in $L$ has two $3^+$-neighbours by construction of $L$, and thus receives $\frac{1}{4}$ from each of its two neighbours by \ref{ninL}. 
Thus every $2$-vertex has non-negative final charge.

Let $v$ be a $3$-vertex, which has initial charge $\frac{1}{2}$. 
By Lemma~\ref{l_3-}, $v$ has at most two $2$-neighbours.
Thus, if $v$ is not in $L$, then it gives no charge by \ref{inL} and gives $\frac{1}{4}$ at most twice by \ref{ninL}. 
Suppose $v$ is in $L$.
By Lemma~\ref{l_3v}, $v$ has at most one son in $L$.
By construction of $L$, if it has one son and another $2$-neighbour $u$, then $u$ is the father of $v$. 
Therefore, $v$ has either at most one $2$-neighbour or two $2$-neighbours where one is its father in $L$.
Consequently, it gives $\frac{1}{2}$ to its son by \ref{inL} and gives nothing by \ref{ninL}. 
In all cases, $v$ has non-negative final charge.

Let $v$ be a $4$-vertex, which has initial charge $\frac{3}{2}$.
By Lemma~\ref{l_4v}, $v$ has at most three sons in $L$. 
If $v$ has at most two sons, then it gives charge at most $2\times\frac{1}{4}+2\times\frac{1}{2} = \frac{3}{2}$ by the discharging rules. 
If $v$ has three sons and the last neighbour is not a 2-vertex, then it gives charge  $3\times\frac{1}{2} = \frac{3}{2}$ by the discharging rules. 
If $v$ has three sons in $L$ and the last neighbour $u$ is a $2$-vertex, then by construction of $L$, $u$ is the father of $v$.
%Since $v$ cannot have four sons by Lemma~\ref{l_4v}, $u$ is not the son of $v$.
%Moreover, 
%By construction, $u$ is the father of $v$, which implies that $u$ is in $L$.
% hence if $u$ is not the father of $v$, then $u$ is not in $L$ ($u$ cannot be a leaf of $L$ since $v$ is a $3^+$-vertex), and we could apply Step 2 of the construction of $L$. Thus $u$ is the father of $v$. Therefore $v$ has at most two sons in $L$, in which case 
Thus, $v$ gives charge $3\times \frac{1}{2}= \frac{3}{2}$ by \ref{inL} and gives nothing by \ref{ninL}.
In all cases, $v$ has non-negative final charge.

Each $5^+$-vertex with degree $d$ has initial charge $d - \frac{5}{2} \ge \frac{d}{2}$, and gives charge at most $\frac{d}{2}$ by the discharging rules. Therefore every $5^+$-vertex has non-negative final charge.
\end{proof}

\section{Proof of Theorem~\ref{t32_main}}\label{s_t32_main}
%Let us first recall the statement of Theorem~\ref{t32_main}.
%
%\begingroup
%\def\thetheo{\ref{t32_main}}
%\begin{theo}[recall]
%Let $k \ge 2$ be an integer. Every graph $G$ with $\mad(G) < \frac{8k}{3k+1} = \frac{8}{3}\left( 1 - \frac{1}{3k+1} \right)$ admits an $({\cal I},{\cal O}_k)$-partition.
%\end{theo}
%\addtocounter{theo}{-1}
%\endgroup

Let $G$ be a counterexample to Theorem~\ref{t32_main} with as few vertices as possible. 
%Our discharging procedure consists in giving weights to vertices so that the total sum of the weights is negative and then in moving these weights according to some rules from vertex to vertex so that all the vertices eventually have non-negative weights (due to some structural properties of $G$). That contradiction completes the proof.
The graph $G$ is connected, since otherwise at least one of its components would be a counterexample with fewer vertices than $G$. 
This further implies that every vertex of $G$ has degree at least 1. 

Given an (partial) $({\cal I},{\cal O}_k)$-partition of $G$, we say a vertex $v$ with colour $O$ is \emph{saturated} if the component in $G[O]$ containing $v$ is of order $k$. 
Let us first prove an easy lemma on the structure of $G$.

\begin{lemm} \label{l32_co}
%The graph $G$ is connected and its minimum degree is at least $2$.
Every vertex in $G$ has degree at least $2$.
\end{lemm}

\begin{proof}
%If $G$ is not connected, then one of its components is a smaller counter-example to Theorem \ref{t32_main}, contradicting the minimality of $G$.
%
%If $G$ contains a $0$-vertex, then $G$ is reduced to that vertex and admits an $({\cal I},{\cal O}_k)$-partition.
%
Let $v$ be a vertex of degree 1 in $G$.
Since the graph $G - v$ has fewer vertices than $G$, it admits an $({\cal I},{\cal O}_k)$-partition, which can be extended to an $({\cal I},{\cal O}_k)$-partition of $G$ by giving to $v$ the colour distinct from that of its neighbour. 
This contradicts that $G$ is a counterexample. 
%
%Let $v$ be a $1$-vertex of $G$. The graph $G-v$ has fewer vertices than $G$, therefore it admits an $({\cal I},{\cal O}_k)$-partition. If the neighbour of $v$ is coloured $I$, then we put $v$ in $O$, and otherwise we put $v$ in $I$. 
%This leads to an $({\cal I},{\cal O}_k)$-partition of $G$, which is a contradiction.
\end{proof}

The key of our proof is the use of two procedures, Procedure~\ref{p_gen} (the main procedure) that calls Procedure~\ref{p_sub} (the  secondary inductive procedure). 
Procedure~\ref{p_gen} starts with a colouring that is almost an $({\cal I},{\cal O}_k)$-partition of $G$. From that colouring, we try to obtain an $({\cal I},{\cal O}_k)$-partition of $G$, leading to some structural properties of $G$. 
The revealed structure will enable us to apply discharging on $G$, concluding that a counterexample to the theorem cannot exist. 

For $i \in \{0,1,2\}$, let $V^i$ (resp $I^i$, $O^i$) be the set of vertices of $V(G)$ (resp. $I$, $O$) marked $i$ times.
During the procedures, no vertex is marked in the beginning and each vertex of $G$ will be {marked} at most twice. 
Some vertices may be marked and then unmarked. 
The general principle is that a vertex $v$ is marked once when we first encounter $v$, and a second time when we actually treat $v$. 
In the end, each vertex will be either marked twice or unmarked; a vertex is marked once only temporarily.
The vertices marked twice will be able to give weight to their unmarked neighbours. 
For the rest, the weight will move between vertices marked twice. 

In Procedures \ref{p_sub} and \ref{p_gen}, we will define the following notions, which will be useful to describe our discharging procedure. 
For some vertices of $G$, we will define the \emph{cluster set} $S(v)$ of a vertex $v$ to be a set of vertices such that $v\notin S(v)$, $S(v) \cup \{v\}$ induces a connected graph. 
Given a vertex $v$, an element of $S(v)$ is a \emph{subordinate} of $v$. 
We say that $v$ \emph{supervises} the elements of $S(v)$. 
We will also define a cluster set with no supervisor. 
For some vertices of $G$, a \emph{giving edge}, denoted $e(v)$, is a fixed edge incident to $v$. 
Finally, some edges of $G$ will become \emph{neutral edges}. 
Some ideas to keep in mind are the following:
\begin{itemize}
\item We will specify the discharging procedure in each cluster set.
\item Only vertices marked twice will give weight.
\item A vertex marked twice will give less weight to its adjacent subordinates than to its unmarked neighbours.
\item A giving edge will link two vertices marked twice. 
A giving edge $e(v) = vw$ of $v$ means that $w$ will be able to give weight to $v$ as if $v$ is unmarked (at some point of the procedures, $w$ is marked twice and $v$ is not marked yet, and thus $w$ could give weight to $v$).
\item Neutral edges are edges between vertices that do not give weight to each other.
\end{itemize}

Procedures \ref{p_sub} and \ref{p_gen} try to recolour some of the vertices.
When we recolour some vertices, we remember the colour they had, so that we can revert all the recolourings that have been done since the start of the procedure.
When we say that we revert all the recolourings that have been done since the start of the procedure, or that we reset the colourings, we mean since the start of the call of the procedure that we are currently treating. This is important since the procedure calls itself recursively. For example, suppose we make a call $A$ to the procedure, and that call makes a call $B$ to the procedure. If we say in the code of $A$ before the call to $B$ that we reverse all the recolourings that have been done, then we undo all the recolourings done since the start of call $A$. If we say this in the code of $B$, then we only reverse what has been recoloured since the start of call $B$. Lastly, if we say this in the code of $A$ after the call to $B$, then we reverse all the changes that were done since the beginning of call $A$, including those that were done in call $B$. 
Similarly, we will also reverse all the markings that have been done since the beginning of the procedure. When we do so, we also reset the cluster sets and the giving and neutral edges as they were at the start of the procedure.

\medskip

Procedure \ref{p_sub} assumes that all vertices of $G$ are coloured with either $O$ or $I$, such that no two adjacent vertices have colour $I$; no assumption is made on vertices with colour $O$. 
Considering an $O^2$-vertex $v$, the aim of Procedure \ref{p_sub} is either to recolour $v$ with colour $I$, or to determine a cluster set or a neutral edge for $v$.
The sets $W$, $C$, and $C'$ in Procedure \ref{p_sub} are not sets defined for the whole procedure, but local sets defined at each call of the procedure.
Let us now present Procedure \ref{p_sub} that is applied to an $O^2$-vertex $v$:

\begin{proc}\label{p_sub}Let $W$ be the set of $I$-neighbours of $v$.
\begin{enumerate}
\item  If a vertex $w$ of $W$ is marked twice, then we set $S(v) = \emptyset$, we set $vw$ as a neutral edge, and we end the procedure. \label{step1}
\end{enumerate}

Now every vertex of $W$ is marked at most once. 
\begin{enumerate}[resume]
\item 
If a vertex $w$ of $W$ has an $O^2$-neighbour $u$ distinct from $v$, then we set $S(v) = \{w\}$ and $e(w) = wu$, we mark $w$ twice, and we end the procedure.\label{step2}
\item 
If a vertex $w$ of $W$ has an $O^1$-neighbour $u$, then we set $S(v) = S(u) = \{w\}$, we mark $u$ a second time, we mark $w$ twice, and we end the procedure. \label{step2bis}
\item 
All vertices of $W$ have only $O^0$-neighbours. We recolour $v$ with $I$ and every vertex of $W$ with $O^0$. Let $X$ be the set of the vertices of $W$ that are in a component of $G[O]$ with at least $k+1$ vertices. Let $C$ be the set of vertices of the components of $G[O^0]$ containing (at least) a vertex of $X$.

While there exists a vertex in $C \setminus X$ that has all its neighbours in $O$, put that vertex in $I$, and redefine $X$ and $C$ as before (but with respect with the new colouring of the graph).

Mark all the $3^+$-vertices of $C \setminus X$ once, and mark the vertices of $X$ twice. 

Now, while there is a vertex $c \in C\setminus X$ marked once, do the following:
	\begin{enumerate}
	\item Mark $c$ a second time and apply Procedure~\ref{p_sub} to the vertex $c$. \label{step3a}

	\item If $c$ is in $I$, then we undo all the markings that have been done since the start of the procedure, reset the cluster sets and giving and neutral edges as they were at the start of the procedure,  redefine $X$ and $C$ as before (but with respect to the new colouring of the graph), mark all the $3^+$-vertices of $C \setminus X$ once, and mark all the elements of $X$ twice. \label{step3b}
	
	\end{enumerate} \label{step3}
	
\item If $X$ is empty at the beginning of 4. or becomes empty during 4.(b), then we stop the procedure (in that case, we succeed to recolour $v$ with $I$ and managed all the resulting conflicts).

Now, set $S(v) = \emptyset$.

For every component $C'$ of $G[C]$ with at least $k+1$ vertices, we add to $S(v)$ the $3^+$-vertices of $C'$, plus the $2$-vertices of $C' \cap X$.

For every component $C'$ of $G[C]$ with at most $k$ vertices such that a $3^+$-vertex $x$ in $C'$ is adjacent to a vertex $y$ that was already marked twice at the start of the procedure, we add to $S(v)$ the $3^+$-vertices of $C'$, plus the $2$-vertices of $C' \cap X$, and set $e(x) = xy$ for one such $x$.

We undo all the recolourings that have been done since the beginning of the procedure. \label{step4}
\end{enumerate}
\end{proc}

We now prove a series of properties of Procedure~\ref{p_sub}.

\begin{lemm} \label{sub_unmark}
In a given call to Procedure~\ref{p_sub}, the only vertices that become unmarked are vertices that have been marked during the procedure.
\end{lemm}

\begin{proof}
Every time we undo the markings that have been done in the procedure, the marking state resets to what it was at the start of the procedure. Since we never unmark a vertex aside from those resets, no vertex that was marked at least once at the start of the procedure becomes unmarked in the procedure.
\end{proof}

By a \emph{single call}, we mean a call without considering the inner recursive calls.
When we say that a change (of the colouring or the markings, for example) happens \emph{directly}, we mean it does not happen during a reset of the colouring or the markings of the vertices.
The following lemma is straightforward.

\begin{lemm} \label{sub_OI}
In a single call to Procedure~\ref{p_sub} applied to vertex $v$, the only time a vertex is recoloured from $O$ to $I$ is by undoing some recolourings (end of Step~\ref{step4}), except for $v$ and some vertices in $O^0$ with no neighbours in $I$ during Step~\ref{step3}. 
\end{lemm}

\begin{lemm} \label{sub_dis}
At the end of a call to Procedure~\ref{p_sub} applied to vertex $v$, either all the vertices have the same colour as before the call, or they have all the same markings, subordinates, giving edges, and neutral edges as before the call. Moreover, at the end of Procedure~\ref{p_sub}, $v$ is in $O$ if and only if all the vertices have the same colour as before the call.
\end{lemm}

\begin{proof}
In Steps~\ref{step1}, \ref{step2}, and~\ref{step2bis}, there are no recursive calls and no recolourings, therefore if the procedure stops in one of those, then all the vertices have the same colour as before the call. In those cases, the vertex $v$ is still in $O$.
If Step~\ref{step3} is reached, then Step~\ref{step4} is also reached, and either all the recolourings are undone (end of Step~\ref{step4}), or all the markings, subordinates, and giving edges are reset (Step \ref{step3}.b) or unchanged (the case where $X$ is empty). As $v$ is given colour $I$ in Step~\ref{step3}, if the colourings are not undone, then $v$ is in $I$. If all the recolourings are undone, then $v$ is back in $O$. This ends the proof of the lemma.
\end{proof}

In the following, we will generally ignore what was done then undone. In particular, by Lemma~\ref{sub_dis}, we can assume that every call to Procedure~\ref{p_sub} does not change either the colouring or the markings, subordinates, giving edges, and neutral edges of the vertices.

The following two lemmas ensure that the procedure keeps the colouring at least as good as it was.

\begin{lemm} \label{sub_no2I}
If no two vertices of $I$ are adjacent before a call to Procedure~\ref{p_sub}, then no two vertices of $I$ are ever adjacent during the call.
\end{lemm}

\begin{proof}
By contradiction, let $v$ and $w$ be two adjacent vertices that are not both in $I$ and that become adjacent $I$-vertices during a call to Procedure~\ref{p_sub}.
Let us consider the first time $v$ and $w$ are both in $I$. This happens after one of $v$ and $w$, say $v$, was recoloured from $O$ to $I$.
That recolouring cannot be the result of a reset of the colourings, since $v$ and $w$ were never in $I$ together before. Thus $v$ was recoloured to $I$ while its neighbour $w$ was in $I$. By Lemma~\ref{sub_OI}, since $v$ has a neighbour in $I$, this implies that $v$ was recoloured as the vertex for which some call $B$ to Procedure~\ref{p_sub} was applied. In call $B$, $v$ was recoloured from $O$ to $I$ in Step~\ref{step3}, and $w$, that was in $W$ at this point as an $I$-neighbour of $v$, was recoloured from $I$ to $O$ at the same time. Therefore $v$ and $w$ were never in $I$ at the same time.
\end{proof}

\begin{lemm} \label{sub_col}
After a call to Procedure~\ref{p_sub}, the vertices of all the components of $G[O]$ with more than $k$ vertices were in $O$ at the start of the call.
\end{lemm}

\begin{proof}
Consider a call to Procedure~\ref{p_sub} for which the lemma is not true. We assume by induction that the recursive calls verify the lemma.
By Lemma~\ref{sub_dis}, we can assume that we are in a call for a vertex $v$ that is in $I$ at the end of the call (otherwise the colouring is the same as before the call). This implies that in Step~\ref{step4}, the set $X$ was empty, and thus all the elements of $W$ were not in components of $G[O]$ with more than $k$ vertices. The only time when vertices are put from $I$ to $O$ in Procedure~\ref{p_sub} is when we recolour the vertices of $W$ in the beginning of Step~\ref{step3}. This proves the lemma.
\end{proof}

The next two lemmas prove that Procedure~\ref{p_sub} terminates.

\begin{lemm} \label{sub_Cdec}
In Step~\ref{step3} of Procedure~\ref{p_sub}, every time the set $C$ changes, it ends up with fewer vertices than before.
\end{lemm}

\begin{proof}
Let $A$ be the current call to Procedure~\ref{p_sub}.
Assume that the set $C$ of call $A$ has changed after a call $A'$ to Procedure~\ref{p_sub} for some vertex $c \in C$. By the definition of $C$ in Step~\ref{step3b} and by Lemma~\ref{sub_dis}, this implies that $c$ has changed colour from $O$ to $I$. The vertex $c$ is no longer in $G[O^0]$, thus it is removed from $C$.

Let us prove that no new vertex was added to $C$. Consider a vertex that is in $C$ after a change of $C$ and was not in $C$ the previous time $C$ was defined. Note that the state of the markings just before the change of $C$ is the same as at the start of call $A$, either because they did not change or because they were reset. Also note that the only time a vertex may have been recoloured from $I$ to $O$ is in a recursive call called by $A$.
As the elements of $X$ correspond to vertices of $W$ that are in components of $C$ with at least $k+1$ vertices, by Lemma~\ref{sub_col} no new vertex is put in $X$. Furthermore, the elements of $X$ that remained in $X$ do not have any new vertex in their component of $G[O]$, and thus, since the markings are the same, no new vertex in their component of $G[O^0]$ either. Therefore $C$ does not have any new vertex.
This ends the proof of the lemma.
\end{proof}

\begin{lemm} \label{sub_once}
Procedure~\ref{p_sub} terminates, and at the end of the procedure, every vertex that is marked once was marked once at the start of the procedure.
\end{lemm}

\begin{proof}
Procedure~\ref{p_sub} calls itself recursively. To make sure that it terminates, we only need to make sure on the one hand that such calls cannot be infinitely nested, and on the other hand that each of these calls only makes a finite number of calls to the procedure.
First, note that every time we call Procedure~\ref{p_sub} recursively, at least one additional vertex is marked twice (the vertex for which the new call is applied). Therefore there cannot be infinitely nested inductive iterations. 

Now we proceed by induction. Assume that for a given call $A$ to the procedure, each recursive call stops while not marking once any vertex that was not marked once. 
If the procedure stops in Step~\ref{step1},~\ref{step2} or~\ref{step2bis} then clearly no new vertex is marked once. Now we assume that Step~\ref{step3} is reached. Let $n_1$ be the number of vertices in $C$, and let $n_2$ be the number of vertices that are marked once in $C$. 

By Lemma~\ref{sub_Cdec}, every time $C$ changes, $n_1$ decreases strictly. When it does not change during an iteration of Step~\ref{step3b}, then at least one vertex of $C$ that was marked once is now marked twice, thus $n_2$ decreases strictly.
At each iteration of Steps~\ref{step3a} and~\ref{step3b}, the lexicographical order $(n_1,n_2)$ decreases strictly, therefore there are finitely many such iterations. Moreover, whenever the set $C$ changes, all the markings made during the procedure are cancelled, and in the end, all the vertices that have been marked once (only vertices of $C$, by induction hypothesis), are now marked twice. In Step~\ref{step4}, there are no new vertices that are marked once, and the procedure stops.
\end{proof}

\begin{lemm} \label{sub_noI1}
Throughout the procedure, no vertex is ever in $I^1$.
\end{lemm}

\begin{proof}
Let $v$ be the first vertex that is put in $I^1$. We consider the call $A$ where $v$ was first in $I^1$. As we never alter the colourings and markings at the same time, either $v$ was put from $O$ to $I$ while marked once, or it became marked once while in $I$.

Suppose $v$ was put from $O$ to $I$ while marked once.  If call $A$ was applied for vertex $v$, then $v$ was marked twice at the start of call $A$, and by Lemmas~\ref{sub_unmark} and~\ref{sub_once} $v$ remained  marked twice throughout call $A$, a contradiction. Therefore call $A$ was not applied for vertex $v$, and thus by Lemma~\ref{sub_OI}, the recolouring from $O$ to $I$ of $v$ was done during a reset of the recolouring, and thus at the end of Step~\ref{step4} for call $A$. Therefore $v$ was marked once at the end of call $A$. This implies by Lemma~\ref{sub_once} that $v$ was already marked once at the start of call $A$. Moreover, since $v$ was recoloured from $O$ to $I$ during the reset of the colouring in Step~\ref{step4}, it was in $I$ at the start of call $A$. Therefore $v$ was in $I^1$ at the start of call $A$, a contradiction.

Suppose $v$ became marked once while in $I$. The only vertices that we mark once directly are vertices of $C$, that are in $O$. Therefore $v$ became marked once due to a reset of the markings, thus $v$ was marked once at the start of call $A$, thus it was in $O^1$. By Lemma~\ref{sub_unmark}, $v$ never became unmarked until the end of call $A$. By Lemma~\ref{sub_OI}, the first time $v$ was put in $I$, $v$ was the vertex that a certain call $B$ of Procedure~\ref{p_sub} was applied for. But this implies that when the call $B$ was applied, $v$ was in some set $C$, and thus had been unmarked at some point since the start of call $A$, which leads to a contradiction.
\end{proof}

\begin{lemm} \label{sub_Smark}
Every time Procedure~\ref{p_sub} changes the cluster set of a vertex, if this change is not undone afterwards during the procedure, then that vertex either is the vertex we are applying the procedure for, or was not marked twice at the start of the procedure and is  marked twice at the end of the procedure.
\end{lemm}

\begin{proof}
We consider a call $A$ to the procedure.
Let $v$ be a vertex whose cluster set is changed during call $A$ or one of its recursive calls. If $v$ is the vertex some call $B \ne A$ is applied for, then it became marked twice just before call $B$, since calls to Procedure~\ref{p_sub} are always made for vertices that have just become  marked twice. By Lemmas~\ref{sub_unmark} and~\ref{sub_once}, since we only mark once vertices that were previously unmarked, $v$ was not marked twice at the start of call $A$ and it remains  marked twice until the end of call $A$. The only time the procedure changes the cluster set of another vertex is when $u$ was marked once in Step~\ref{step2}, in which case the vertex $u$ is also marked a second time. Similarly, by Lemmas~\ref{sub_unmark} and~\ref{sub_once}, and since we only mark once vertices that were previously unmarked, $u$ was not  marked twice in the beginning of the procedure. 
As the cluster sets and the markings are always reset together, the lemma follows.
\end{proof} 

Note that this implies that as long as we only apply Procedure~\ref{p_sub} to vertices that just became  marked twice (as we do in Step~\ref{step3a}), we will never set  the cluster set of a vertex twice.

Let us define a property $\Pi$: 

\begin{prop*}[$\Pi$]
No $3^+$-vertex in $O^0$ has an $O^1$-neighbour, and no $2$-vertex in $O^0$ has a neighbour in $O^1$ and the other neighbour in $O$.
\end{prop*}

That property is ensured by the fact that when we put vertices in $O^1$, all the vertices of the component in $O^0$ are put there together, except from some $2$-vertices that have an $I$-neighbour. 

\begin{lemm} \label{sub_Pi}
If $\Pi$ is true before a call $A$ to Procedure~\ref{p_sub} until the end of call $A$, every time $\Pi$ becomes false, it is restored before any recursive call is applied, and before the procedure ends.
\end{lemm}

\begin{proof}
We consider the first time $\Pi$ is not true at the start or the end of some call to Procedure~\ref{p_sub}. We consider the event that caused $\Pi$ to fail.
By induction, we suppose that every recursive call preserves $\Pi$. 
Let us consider the change that caused the property $\Pi$ to be broken, and prove that $\Pi$ is restored immediately.

Suppose the change was direct, i.e. not because of some reset. When we mark a vertex $u$ once, which only happens in Step~\ref{step3}, we do so for all the vertices in the component of $u$ in $G[O^0]$ except some $2$-vertices that have an $I$-neighbour, and thus $\Pi$ is maintained. Moreover, by Lemma~\ref{sub_unmark}, no vertex directly becomes unmarked. Therefore the change was a change in the colouring.
Hence some vertex $w$ was put from $I$ to $O$. Moreover, as no vertex that is marked once ever changes colour directly, $w$ was unmarked when it changed colour. Therefore either $w$ has a neighbour $u$ in $O^1$, or it has a $2$-neighbour $u$ in $O^0$ that has a neighbour $x$ in $O^1$. Vertex $w$ is put in $O$ as a vertex of $W$ in Step~\ref{step3} of some call $B$ to Procedure~\ref{p_sub}. If call $B$ was applied for vertex $u$, then $u$ is put in $I$ when $w$ is put in $O$, and the property is not broken. Therefore call $B$ was not applied for $u$. Then $u$ is not in $O^1$, otherwise we would have stopped call $B$ in Step~\ref{step2}. Thus $u$ is a $2$-vertex in $O^0$ incident to $w$ with two $O$-neighbours, and thus it is put in $I$ in Step~\ref{step3} of call $B$.

Now we suppose the considered change is some reset. Suppose first that the reset was a reset of the colouring. This happens only at the end of Step~\ref{step4}, therefore at the end of the procedure, and the colouring is reset to what it was at the start of the procedure. By Lemmas~\ref{sub_unmark} and~\ref{sub_once}, no new vertex is marked once or zero times at the end of the procedure compared to the start, so the property was already broken at the start of the current call to the procedure, a contradiction.

Thus the reset was a reset of the markings, subordinates and giving edges. Then some vertex $w$ must have been put from $I$ to $O$ between the start of the procedure and this reset was such that $w$ has an $O$-neighbour $u$ that either was in $O^1$ at the start of the procedure, or is a $2$-vertex that was in $O^0$ at the start of the procedure, such that the other neighbour $x$ of $u$ was in $O$ at the start of the procedure. Since the resets are well parenthesised, we may assume that the recolouring of $w$ was not the consequence of a reset, i.e. was done directly. We are now interested in the state of the procedure at the time when $w$ is put from $I$ to $O$. Note that this recolouring must happen in Step~\ref{step3} for some call $B$, and that the marking in the beginning of Step~\ref{step3} is still the same as at the start of call $B$. Also note that, unless $u$ is a $2$-vertex with its two neighbours in $O$, $u$ cannot be in $O^0$ by Lemma~\ref{sub_unmark}. If call $B$ was applied for vertex $u$, then $u$ is put in $I$ when $w$ is put in $O$, and the property is not broken. Therefore call $B$ was not applied for $u$. Then $u$ is not in $O^1$ nor $O^2$, otherwise we would have stopped call $B$ in Step~\ref{step2}. Thus $u$ is a $2$-vertex in $O^0$ incident to $w$ with two $O$-neighbours, and thus it is put in $I$ in Step~\ref{step3} of call $B$.
\end{proof}

The following lemma implies that Procedure~\ref{p_sub} either manages to recolour vertex $v$ with colour $I$, or gives it a neutral edge or a cluster set.

\begin{lemm} \label{sub_vgood}
Let $A$ be a call to Procedure~\ref{p_sub}. Suppose that before a call $A$ to Procedure~\ref{p_sub}, $\Pi$ is true.
After call $A$, if the vertex $v$ for which the call was applied is in $O$, then either $v$ is incident to a neutral edge%, or $v$ and another vertex $u$ both have a vertex $w$ as a single subordinate
, or $S(v) \ne \emptyset$.
\end{lemm}

\begin{proof}
If Procedure~\ref{p_sub} stops in Step~\ref{step1}, then $v$ is incident to a neutral edge. If Procedure~\ref{p_sub} stops in Step~\ref{step2} or Step~\ref{step2bis}, then $S(v)$ is not empty.
If the procedure reached Step~\ref{step3}, then it reached Step~\ref{step4}. We can assume that the vertex $v$ is in $O$ at the end of the procedure, and thus by Lemma~\ref{sub_dis} that all the recolourings have been undone. 

All that remains to prove is that $S(v) \ne \emptyset$. We know that $X$ is not empty in Step~\ref{step4}, otherwise $v$ would be in $I$. Let $x$ be a vertex of $X$. By Lemma~\ref{sub_noI1}, since $x$ was in $I$ at the start of the procedure (as a member of $W$), it was not marked once. Therefore $x \in C$. By definition of $X$, the component $K$ of $x$ in $G[O]$ has at least $k+1$-vertices. By Lemma~\ref{sub_col}, since $K$ still has at least $k+1$ vertices the last time $X$ is redefined, $K$ did not gain any new vertex since the beginning of the recursive calls. Let $C'$ be the component of $x$ in $G[C]$. If $C'$ has at least $k+1$ vertices, then $x \in S(v)$. Otherwise, by definition of $X$, there is a vertex $y$ in $C'$ that has a neighbour $z$ that was in $O^1$ or $O^2$ last time $X$ was defined. By Lemma~\ref{sub_Pi}, the graph $G$ verified $\Pi$ just before we last marked once all the vertices of $C\setminus X$ in Step~\ref{step3}, and thus the last time $X$ was defined. Therefore $z$ was in $O^2$ the last time $X$ was defined, and thus also the first time $C$ was defined. The vertex $y$ is not in $W$ (otherwise we would have stopped in Step~\ref{step2}), therefore it is not a $2$-vertex (otherwise we would have put it in $I$ as a $2$-vertex in $C$ with two $O$-neighbours). Therefore $x \in S(v)$.
\end{proof}

\bigskip
Let us now define a new procedure, applied on the graph $G$ with no colouring, but where some vertices may be  marked twice. The aim of that new procedure is, for some unmarked $2$-vertex $v$, to make one of its neighbours become  marked twice.

 We start with every vertex of $G$ unmarked. 
While there is a $2$-vertex $v$ that is not  marked  twice and has no neighbour that is  marked twice, we apply the following procedure.

\begin{proc}\label{p_gen}{\ }
\begin{enumerate}
\item The graph $G-v$ has fewer vertices than $G$, therefore it admits an $({\cal I},{\cal O}_k)$-partition. The vertex $v$ has a neighbour $u$ in $I$ and a neighbour $w$ in $O$ that is saturated, otherwise we can colour $v$ with either $I$ or $O$, which leads to a contradiction. 
assign colour $O$ to $v$ (note that this leads to a component of order $k+1$ in $G[O]$).
\label{step1'}

\item Let us consider the component $C$ of $v$ in $G[O^0]$. 
Mark once every $3^+$-vertex of $C$. For every once marked vertex $x$ of $C$, we mark $x$ a second time and we apply Procedure~\ref{p_sub} for $x$. Note that there is no $2$-vertex in $C$ with both neighbours in $O$, as otherwise we could put such a vertex in $I$ and get a colouring of $G$. We define a cluster set $S$ as the set of $3^+$-vertices in $C$. The set $S$ is defined with no supervisor, and if there is an edge from a vertex $x$ of $S$ to a vertex $y$ that was  marked  twice at the beginning of Procedure~\ref{p_gen}, then let $e(x) = xy$ for one such $x$.
\label{step2'}
\end{enumerate}
\end{proc}

We are now ready to prove some more lemmas, to ensure that the graph obtained after all the calls to Procedure~\ref{p_gen} have been applied verifies the properties we will need in the discharging procedure.

\begin{lemm} \label{gen_col}
In Procedure~\ref{p_gen}, no vertex of $C$ is ever put in $I$, and at the end of each call to Procedure~\ref{p_sub}, all the recolourings have been undone.
\end{lemm}

\begin{proof}
By Lemma~\ref{sub_col}, the calls to Procedure~\ref{p_sub} never add any new vertex to the component of $v$ in $G[O]$, otherwise this component ends up with at most $k$ vertices, and by Lemmas~\ref{sub_no2I} and~\ref{sub_col}, we get an $({\cal I},{\cal O}_k)$-partition of $G$, a contradiction.

Procedure~\ref{p_gen} cannot lead to recolouring to $I$ any of the vertices of $C$. Indeed, by Lemmas~\ref{sub_no2I} and~\ref{sub_col}, this would lead to an $({\cal I},{\cal O}_k)$-partition of $G$. Thus by Lemma~\ref{sub_dis}, in every application of Procedure~\ref{p_sub}, all the recolourings have been undone.
\end{proof}

\begin{lemm} \label{gen_once}
After any number of calls to Procedure~\ref{p_gen}, no vertex is marked once.
\end{lemm}

\begin{proof}
The property is true initially, when every vertex is unmarked.
Suppose that no vertex is marked once at the beginning of some call to Procedure~\ref{p_gen}. We mark each $3^+$-vertex in $C$ once. Afterwards, for each vertex $x \in C$ that was marked once, we mark it twice and apply Procedure~\ref{p_sub}, that does not mark once any new vertex by Lemma~\ref{sub_once}. Therefore at the end of the call, no vertex is marked once. The lemma follows by induction.
\end{proof}
 
 \begin{lemm} \label{gen_P}
After any number of calls to Procedure~\ref{p_gen},
every  vertex marked twice is either adjacent to one of its supervisors, or has a neutral edge, or has a non-empty cluster set.
\end{lemm}

\begin{proof}
Let $w$ be a vertex marked twice.
Initially, no vertex is marked twice. Therefore $w$ became  marked twice at some point, and this change was not reset afterwards. Let us consider the moment $w$ was marked twice.
If $w$ is  marked twice in Step~\ref{step2} or~\ref{step2bis} of Procedure~\ref{p_sub}, then either it is adjacent to one of it supervisor, or it is given a subordinate. If $w$ is  marked twice in Step~\ref{step3} of Procedure~\ref{p_sub} as a vertex of $W$, it later becomes adjacent to its supervisor $v$ in Step~\ref{step4}.

Suppose $v$ was marked twice  in Procedure~\ref{p_gen} or in Step~\ref{step3a} of Procedure~\ref{p_sub}. Then some call $A$ to Procedure~\ref{p_sub} is applied for this vertex right after. Therefore, by Lemma~\ref{sub_vgood}, since $\Pi$ is true at the start of each call to Procedure~\ref{p_sub}, at the end of call $A$, either $v$ has a neutral edge, or it has a non-empty cluster set, or it is in $I$. If it is in $I$, then by Lemma~\ref{gen_col}, $v$ was marked twice  in Step~\ref{step3a} of Procedure~\ref{p_sub}, and thus the marking of $v$ is reset in Step~\ref{step3b}, a contradiction. Therefore $v$ has a neutral edge or a non-empty cluster set. That ends the proof of the lemma.
\end{proof}
 
\begin{lemm} \label{noMnosub}
No unmarked vertex has a cluster set or a giving edge.
\end{lemm}

\begin{proof}
This is true in the beginning as no vertex has a cluster set or a giving edge. Every time a vertex is given a cluster set or a giving edge in Procedures~\ref{p_sub} and~\ref{p_gen}, this vertex is marked twice. We always reset the cluster sets, the giving edge and the markings together.  Finally, no vertex marked twice directly becomes unmarked or marked once.
\end{proof}

\begin{lemm} \label{gen_inter}
Every vertex belongs to at most two cluster sets. 
Moreover if a vertex belongs to two cluster sets, then that vertex has no giving edge and the two sets contain only that vertex.
\end{lemm}

\begin{proof}
Let $v$ be a vertex in two cluster sets $S$ and $T$. When $S$ was set, $v$ was marked twice. Indeed, every vertex that is put in a cluster set is marked twice. As the cluster sets are reset together with the markings, and no vertex marked twice directly becomes unmarked or marked once, this implies that $v$ remains marked twice  afterwards.

Let us prove that $v$ was unmarked at the start of the call to Procedure~\ref{p_sub} or~\ref{p_gen} in which $S$ is set. If $S$ is set in Steps~\ref{step1},~\ref{step2} or~\ref{step2bis} of Procedure~\ref{p_sub}, this is trivial. If $S$ is set in Step~\ref{step4} of Procedure~\ref{p_sub}, then the last time $C$ was set in Step~\ref{step3}, $v$ was put in $C$, and therefore $v$ was in $O^0$. Furthermore, at that point the set of the markings was the same as it was at the start of the procedure. Therefore if $S$ is set in Procedure~\ref{p_sub}, then $v$ was unmarked at the start of the procedure.
If $S$ is set in Procedure~\ref{p_gen}, then $v$ is in $C$, and thus was unmarked at the start of the procedure.

Therefore $v$ was unmarked at the start of the call to Procedure~\ref{p_sub} or~\ref{p_gen} in which $S$ was set, and was marked twice  at the end of that call. The same holds for $T$. Therefore $S$ and $T$ were set in the same call to Procedure~\ref{p_sub} or~\ref{p_gen}. The only time this can happen is in Step~\ref{step2} of procedure~\ref{p_sub}, and this implies that $S = T = \{v\}$, and $v$ has no giving edge by Lemma~\ref{noMnosub}. By what precedes, if $v$ is also in a third cluster set $U$, then $U$ was also set in the same call to Procedure~\ref{p_sub} or~\ref{p_gen} as $S$ and $T$, which is impossible.
\end{proof}

Note that the previous proof also implies the following lemma.
\begin{lemm} \label{inter_mark}
Every vertex that is in a cluster set is marked twice.
\end{lemm}

\begin{lemm} \label{mark_inter}
After any number of calls to Procedure~\ref{p_gen},
every vertex marked twice is in a cluster set.
\end{lemm}

\begin{proof}
In Procedures~\ref{p_sub} and~\ref{p_gen}, whenever a vertex that was in $V^2$ is put in $V^0 \cup V^1$, it happens because of a reset of the markings. Therefore, since the resets are well nested (they correspond to the nested calls to Procedures~\ref{p_sub} and~\ref{p_gen}), the resets of the markings can never put a vertex from $V^0 \cup V^1$ to $V^2$.

When a vertex $w$ that was unmarked becomes marked twice in Procedure~\ref{p_sub} or~\ref{p_gen}, it either happens in Steps~\ref{step2} or~\ref{step2bis} of Procedure~\ref{p_sub}, in which case $w$ is put in $S(v)$, or it happens in Step~\ref{step3} of Procedure~\ref{p_sub}, in which case $w$ is in $C$. Besides that, whenever a vertex $w$ becomes  marked twice, $w$ was marked once, so it was in $C$ for some call $A$ to Procedure~\ref{p_sub} or Procedure~\ref{p_gen}.
Hence every time a vertex $w$ becomes  marked twice, either it is in a cluster set or it is in $C$ for some call $A$ to Procedures~\ref{p_sub} or~\ref{p_gen}, in which case, when call $A$ stops, either the vertex $w$ is still in $C$, in which case it is put in a cluster set, or it is no longer in $C$ and the markings were reset since $w$ was marked twice. Since the cluster sets and the markings are reset at the same time, this proves the lemma.
\end{proof}

\begin{lemm} \label{gen_give}
After any number of calls to Procedure~\ref{p_gen},
if a vertex $x$ has a giving edge $e(x) = xw$, then $x$ and $w$ are marked twice, they are not in the same cluster set, neither of them is the subordinate of the other one, and we do not have $e(w) = wx$.
\end{lemm}

\begin{proof}
Whenever a giving edge $xw$ is given to a vertex $x$ in a call $A$ to Procedure~\ref{p_sub} or~\ref{p_gen}, $w$ was already marked twice at the start of call $A$, and $x$ was unmarked at the start of call $A$. This implies that we do not have $e(w) = wx$ unless the considered edge is reset. Moreover, $x$ is put in the cluster set defined at the end of call $A$, and $w$ is not since it was already  marked twice at the start of call $A$. This implies that $x$ and $w$ are not put in the same cluster set. Moreover, by Lemmas~\ref{noMnosub} and~\ref{inter_mark}, $x$ and $w$ were not the subordinates of one another before the start of call $A$. Furthermore, no vertex that was  marked twice before the start of a procedure is ever put in a cluster set or given a subordinate in that procedure unless it has just become  marked twice and Procedure~\ref{p_gen} is applied to it. Therefore $x$ and $w$ were not the subordinates of one another at the end of call $A$. After call $A$, by Lemma~\ref{noMnosub}, $x$ and $w$ are both  marked twice, and thus neither of them is given new subordinates.
 As the cluster sets and the giving edges are always reset together, that proves the lemma.
\end{proof}

\begin{lemm} \label{gen_neutral}
% ajout pour la preuve du lemme 29 :
A neutral edge is incident to two vertices marked twice.
A neutral edge is not incident to a vertex and one of its subordinates.  A neutral edge is not incident to two vertices belonging to a same cluster set. A neutral edge is not a giving edge.
\end{lemm}

\begin{proof} 
Let us consider a neutral edge $vw$. The only time it can have been set as a neutral edge is in  Step~\ref{step1} of a call $A$ to Procedure~\ref{p_sub} for vertex $v$.
 The cluster set of $v$ is empty, therefore $w$ is not a subordinate of $v$. Moreover, $w$ is a vertex in $I^2$ when $A$ is called in some call $B$ to Procedure~\ref{p_sub} or~\ref{p_gen}, and $v$ was marked twice just before $A$ was called. The vertex $v$ is in $C$ for call $B$, and will be put in the set $S$ containing all the element of $C$ (unless $C$ is redefined, in which case the giving edge $vw$ is reset). Let us consider the time when call $A$ is applied.  When the set $C$ (for call $B$) was defined, it contained only elements of $O$, and if after some call $A'$ called by $B$ before $A$, the colouring had changed, then by Lemma~\ref{sub_dis} the set $C$ would have been redefined. Therefore when $A$ is applied, all the elements of $C$ for call $B$ are in $O$, therefore $w$ is not in that set $C$. Therefore $v$ and $w$ are not in the same cluster set. Additionally, when the giving edges are defined in $B$, they are between vertices in $O$, so $vw$ is not defined as a giving edge in $B$. Before call $A$, the vertex $v$ is not marked twice, and after call $A$, the vertices $v$ and $w$ are both  marked twice, thus $vw$ is never defined as a giving edge.
  Moreover, if $B$ is a call to Procedure~\ref{p_sub}, then $v$ is not in the set $W$ of call $B$, therefore $v$ is not adjacent to its supervisor if it has one, and thus $v$ is not a subordinate of $w$.
As the markings, the cluster sets, the giving edges and the neutral edges are reset together, that ends the proof of the lemma.
\end{proof}

\begin{lemm} \label{gen_Sgood}
After any number of calls to Procedure~\ref{p_gen},
every non-empty cluster set $S$ satisfies one of the following: 
\begin{itemize}
\item $S$ is a singleton $\{v\}$ with two supervisors ; moreover $v$ has no subordinates and no giving edges.
\item For each component $K$ of $G[S]$, either $K$ has a vertex with a giving edge, or $\sum_{x \in K} (d(x) - 2) \ge k-1$.
\end{itemize}
\end{lemm}

\begin{proof}
The lemma is true initially as there are no cluster sets. 

%Suppose it is true at the start of a call $A$ to Procedure~\ref{p_gen} and false at the end of the call.

%The property $\Pi$ is true before we mark the vertices of $C$ once in call $A$, as no vertex is marked once by Lemma~\ref{gen_once}, and is maintained true when we mark once the vertices of $C$, since we mark once all the elements of a component of $G[O^0]$ except from the $2$-vertices with an $I$-neighbour.
%Lemma~\ref{sub_Pi} implies that $\Pi$ remains true until the end of call $A$.

Let us consider a set $S$ that does not verify the lemma. The set $S$ is defined in some call $A$ to Procedure~\ref{p_sub} or~\ref{p_gen}. If $S$ is defined in Step~\ref{step1},~\ref{step2}, or~\ref{step2bis} of Procedure~\ref{p_sub}, then the lemma is trivially verified for $S$. 
We can assume that $A$ is either a call to Procedure~\ref{p_gen}, or that Step~\ref{step4} is reached. We consider the set $C$ as it is last defined in call $A$, and if $A$ is a call to Procedure~\ref{p_gen}, then let $X = \emptyset$. 
Let $K$ be the component of $S$ defined from a component $C'$ of $G[C]$ ($K$ is defined from $C'$ by removing some vertices that have degree $1$ in $C'$, and degree $2$ in $G$). 
By construction, either $C'$ contains a $3$-vertex with a giving edge, or $|C'| \ge k+1$. Suppose $|C'| \ge k+1$. For every vertex $u$ of $C' \setminus X$, since $u$ was not put in $I$ in Step~\ref{step3}, $u$ has a neighbour that is not in $C$ (that was in $I$ when $C$ was defined), and thus not in $C'$. For every vertex $u$ of $X$, $u$ is adjacent to its supervisor, which is not in $C'$. Therefore $\sum_{x \in K} (d(x) - 2) = \sum_{x \in C'} (d(x) - 2) = \sum_{x \in C'} (d_{C'}(x) - 1)$, where $d_{C'}(x)$ is the number of neighbours of $x$ in $C'$. The graph $C'$ being connected, $\sum_{x \in C'} (d_{C'}(x) - 1) \ge |C'| - 2 \ge k - 1$. Therefore, the set $S$ verifies the lemma.
 \end{proof}

\begin{lemm} \label{gen_2sup}
Let $v$ be a $2$-vertex in $V^2$. The vertex $v$ is adjacent to all of its supervisors.
\end{lemm}

\begin{proof}
Procedure~\ref{p_gen} does not mark a $2$-vertex twice, except in its calls to Procedure~\ref{p_sub}. Moreover, every time a vertex becomes marked once, that vertex is a $3^+$-vertex, thus only $3^+$-vertices are ever marked once. 

 In Procedure~\ref{p_sub}, if a $2$-vertex was marked twice in Step~\ref{step2} or~\ref{step2bis}, then it is adjacent to its supervisor(s), and if a $2$-vertex was marked twice in Step~\ref{step3}, then it belongs to $X \subseteq W$ and thus is also adjacent to its supervisor. No vertex becomes marked twice anywhere else in Procedure~\ref{p_sub}. As the cluster sets and the markings are reset together, that proves the lemma.
\end{proof}

%%%%%%%%%%%%%%%%%%%%%%%%%%%% OK AFTER %%%%%%%%%%%%%%%%%%%%%%%%%%%%%%%%%%%%%%%%%%%%%%%%%%%%%%%%%%%%%%%%%%%%%%%%%%%%%%%%%%%%%%%%%%%%%%%%%%%%%%%%%%%%%%%%%%%%%%%%%%%%%%%%%%%%%%%%
%%%%%%%%%%%%%%%%%%%%%%%%%%%%%%%%%%%%%%%%%%%%%%%%%%%%%%%%%%%%%%%%%%%%%%%%%%%%%%%%%%%%%%%%%%%%%%%%%%%%%%%%%%%%%%%%%%%%%%%%%%%%%%%%%%%%%%%%%%%%%%%%%%%%%%%%%%%%%%%%%%%%%%%%%%%%%%
%%%%%%%%%%%%%%%%%%%%%%%%%%%%%%%%%%%%%%%%%%%%%%%%%%%%%%%%%%%%%%%%%%%%%%%%%%%%%%%%%%%%%%%%%%%%%%%%%%%%%%%%%%%%%%%%%%%%%%%%%%%%%%%%%%%%%%%%%%%%%%%%%%%%%%%%%%%%%%%%%%%%%%%%%%%%%%
%%%%%%%%%%%%%%%%%%%%%%%%%%%%%%%%%%%%%%%%%%%%%%%%%%%%%%%%%%%%%%%%%%%%%%%%%%%%%%%%%%%%%%%%%%%%%%%%%%%%%%%%%%%%%%%%%%%%%%%%%%%%%%%%%%%%%%%%%%%%%%%%%%%%%%%%%%%%%%%%%%%%%%%%%%%%%%
%%%%%%%%%%%%%%%%%%%%%%%%%%%%%%%%%%%%%%%%%%%%%%%%%%%%%%%%%%%%%%%%%%%%%%%%%%%%%%%%%%%%%%%%%%%%%%%%%%%%%%%%%%%%%%%%%%%%%%%%%%%%%%%%%%%%%%%%%%%%%%%%%%%%%%%%%%%%%%%%%%%%%%%%%%%%%%
%%%%%%%%%%%%%%%%%%%%%%%%%%%%%%%%%%%%%%%%%%%%%%%%%%%%%%%%%%%%%%%%%%%%%%%%%%%%%%%%%%%%%%%%%%%%%%%%%%%%%%%%%%%%%%%%%%%%%%%%%%%%%%%%%%%%%%%%%%%%%%%%%%%%%%%%%%%%%%%%%%%%%%%%%%%%%%

\section*{Discharging procedure}
Let $M = \frac{8k}{3k +1}$. Every vertex $v$ starts with weight equal to  $d(v) - M$. As $M$ is larger than the average degree of $G$, the sum of the weights of the vertices is negative. We will move some weights from vertices to other vertices, without introducing or removing weights, so that every vertex has non-negative weight at the end of the procedure. 
This contradiction will complete the proof.

In the discharging procedure, the weights will move via the edges; in other words, the weights will always move from a vertex to its neighbours. 
%We will make sure that the weight that goes through any given edge is at most $M - 2$. 
We apply the following four steps.

\begin{enumerate}[label={\bf Step \arabic*}]
\item Every vertex $w$ with two supervisors receives $\frac{M-2}{2}$ from each of its supervisors. \label{s1}
%Every vertex $v$ that has exactly one subordinate $w$ such that there is no edge $e(w)$ gives $\frac{M-2}{2}$ to its subordinate. \label{s1}

\item Every vertex that is marked twice gives $M - 2$ to each of its unmarked neighbours.\label{s2}

\item Every vertex $v$ with a giving edge $vw$ receives $M-2$ from $w$. \label{s3}
%For every vertex $v$ that has an edge $e(v) = vw$, $w$ gives $M-2$ to $v$. \label{s3}

\item For every component $C$ of the graph induced by a cluster set, we do the following:

 Let $T$ be a spanning tree of $C$. For each edge $uv$, let $n_{Tuv}$ be the sum of $d(w)-2$ for all vertices $w$ in the component of $v$ in $T-u$. 
 For each vertex $u$ of $T$ and each neighbour $v$ of $u$ in $T$, the vertex $u$ gives $M - 2 - n_{Tuv}(4-\frac{3M}{2})$ to $v$ if this value is positive, unless $v$ is in the same component of $T-u$ as a vertex that has a giving edge (in that case $u$ does not give weight to $v$).\label{s4}
\end{enumerate}

\begin{lemm} \label{l_unmnoch}
No unmarked vertex gives weight to another vertex.
\end{lemm}

\begin{proof}
By contradiction, suppose there is an unmarked vertex $v$ that gives weight to another vertex $u$. 

By Lemma~\ref{noMnosub}, $v$ has no subordinate, so it does not give weight in Step~1 of the discharging procedure. 

Since $v$ is unmarked, $v$ does not give weight in Step~2. 

Suppose $v$ gives weight to $u$ in Step~3. Then $e(u)=uv$. Every time a giving edge is set in Procedures~\ref{p_sub} and~\ref{p_gen}, its other endpoint is twice marked. Moreover the marking and giving edges are always reset together. That implies $v$ is twice marked, a contradiction.

Suppose $v$ gives weight in Step~4. Then $v$ is in a cluster set and thus $v$ is twice marked by Lemma~\ref{inter_mark}, a contradiction.
\end{proof}

\begin{lemm} \label{l_unmarked}
Every unmarked vertex has non-negative weight at the end of the discharging procedure.
\end{lemm}

\begin{proof}
We applied Procedure~\ref{p_gen} until every $2$-vertex is either twice marked or has a neighbour that is twice marked. 
Each unmarked vertex either has degree at least $3 > M$ or has initial weight $2-M$ and receives at least $M-2$ in Step~2 of the discharging procedure. Furthermore, each of these vertices does not give weight by Lemma~\ref{l_unmnoch}. That ends the proof of the lemma.
\end{proof}

\begin{lemm} \label{l_neutral}
No weight is given through a neutral edge.
\end{lemm}

\begin{proof}
This is a direct consequence of Lemma~\ref{gen_neutral}. 
\end{proof}

\begin{lemm}\label{l_sup}
A subordinate does not give weight to its supervisor.
A supervisor does not give weight to its subordinates except maybe in Step~1.
%A subordinate (resp. supervisor) does not give weight to its supervisor (resp. subordinate) except may be in Step~1.
%A vertex does not give weight to its supervisor or subordinates except in Step~1.
\end{lemm}

\begin{proof}
Suppose $v$ is the supervisor of a vertex $w$. By Lemmas~\ref{gen_once}, ~\ref{noMnosub}, and~\ref{inter_mark}, $v$ and $w$ are twice marked and thus do not give weight to each other in Step~2. By Lemma~\ref{gen_give}, they do not give weight to each other in Step~3. By Lemma~\ref{gen_inter}, $v$ and $w$ cannot be both contained in a same cluster set ; hence they do not give weight to each other in Step~4. 

Now, if $w$ gives weight to $v$ in Step~1, then $v$ is in two cluster sets. Then, by Lemma~\ref{gen_inter}, $v$ has no giving edge and is alone in its cluster sets, therefore by Lemma~\ref{gen_Sgood}, $w$ has no subordinates, a contradiction. Thus $w$ does not give weight to $v$ in Step~1.
\end{proof}

\begin{lemm} \label{onlyone}
Every vertex $v$ gives weight to each of its neighbours $w$ at most once in the discharging procedure. Moreover, if $v$ gives weight to $w$, then $w$ does not give weight to $v$.
\end{lemm}

\begin{proof}
Let $v$ and $w$ be two adjacent vertices. 
Suppose that $v$ gives weight to $w$ at some point in the discharging procedure. By Lemmas~\ref{gen_once} and~\ref{l_unmnoch}, $v$ is twice marked.

Suppose first that $v$ gives weight to $w$ in Step~1. By Lemma \ref{l_sup}, $w$ does not give weight to $v$ and $v$ does not give weight to $w$ in Steps 2-4.

%Then by Lemma~\ref{gen_inter}, $v$ has exactly one subordinate, which is $w$, and $w$ has no subordinates ; thus $w$ does not give weight in Step~1. By Lemma~\ref{l_sup}, they do not give weight to each other in Steps 2--\ref{s4}.

Suppose now that $v$ gives weight to $w$ in Step~2. As $w$ is then unmarked, it does not give weight by Lemma~\ref{l_unmnoch}. By Lemma~\ref{noMnosub} (resp. Lemma~\ref{inter_mark}), $v$ does not give weight to $w$ in Step~3 (resp. Step~4). %By Lemma~\ref{inter_mark}, $v$ does not give to $w$ in Step~4.

Suppose that $v$ gives weight to $w$ in Step~3. By Lemma~\ref{gen_give}, $v$ and $w$ are in distinct cluster sets, thus they do not give weight to each other in Step~4. Also by Lemma~\ref{gen_give}, we do not have $e(v) = vw$, thus $w$ does not give weight to $v$ in Step~3.

Lastly, suppose for contradiction that $v$ gives weight to $w$ and $w$ gives weight to $v$ in Step~4. Then $v$ and $w$ are in the same (unique by Lemma~\ref{gen_inter}) cluster set $S$ and in the same component $C$ of $G[S]$. Let $T$ be the tree that was chosen in Step~4 of the discharging procedure for $C$. 

Observe that no vertex in the component of $v$ in $T-w$ nor in the component of $w$ in $T-v$ has a giving edge, and thus that no vertex of $C$ has a giving edge. Also observe that $n_{Tvw} + n_{Twv} = \sum_{x \in C} (d(x) - 2) \ge k-1$ by Lemma~\ref{gen_Sgood}. Since $v$ and $w$ both give weight, we have $M - 2 - n_{Twv}(4-\frac{3M}{2}) > 0$ and $M - 2 - n_{Tvw}(4-\frac{3M}{2}) > 0$. Therefore $2M - 4 - (k-1)(4-\frac{3M}{2}) > 0$, i.e. $M > \frac{8k}{3k+1}$, a contradiction. That ends the proof of the lemma.
\end{proof}

%\begin{lemm}
%The weight given through any edge of $G$ is at most $M-2$.
%\end{lemm}
%
%\begin{proof}
%Let $v$ be a vertex that gives more than $M-2$ to one of its neighbours $w$. By Lemma~\ref{onlyone}, this implies that $v$ gives more than $M-2$ to $w$ in Step~4. Therefore $M - 2 + (n_{vw}-1)(\frac{3M}{2}-4) > M-2$, where $n_{vw} \ge 1$, and thus $\frac{3M}{2}-4 > 0$, which implies that $M > \frac{8}{3}$, a contradiction.
%\end{proof}

\begin{lemm} \label{l_dech}
Every vertex that is twice marked has non-negative weight at the end of the discharging procedure.
\end{lemm}

\begin{proof}
Let $v$ be a $d$-vertex that is twice marked ($d\ge 2$). By Lemma \ref{mark_inter}, $v$ is in a cluster set $S$. By Lemma \ref{gen_Sgood}, 
\begin{itemize}
\item either $S$ is the singleton $\{v\}$ with two supervisors, say $u$ and $w$, and no giving edge,
\item or the component of $G[S]$ containing $v$, say $C$, has a vertex having a giving edge or satisfies  $\sum_{x \in C} (d(x) - 2) \ge k-1$.
\end{itemize} 

Consider the former case. The vertices $u$ and $w$ both give $\frac{M-2}{2}$ to $v$ in Step~1. The vertex $v$ gives nothing to $u$ and $w$ by Lemma \ref{onlyone}, and at most $(d(v)-2)(M-2) \le d(v) -2$ to its other neighbours. It follows that $v$ has weight at least equal to $d(v) - M + 2\frac{M-2}{2} - (d(v) - 2) = 0$.

Consider the latter case.
%
%Suppose $v$ is the single subordinate of two vertices $u$ and $w$, {\bf and $v$ has no giving edges}. The vertices $u$ and $w$ both give $\frac{M-2}{2}$ to $v$ in Step~1. The vertex $v$ gives nothing to $u$ and $w$ by Lemma \ref{onlyone}, and at most $(d(v)-2)(M-2) \le d(v) -2$ to its other neighbours. It follows that $v$ has weight at least equal to $d(v) - M + 2\frac{M-2}{2} - (d(v) - 2) = 0$.
%
%
%Suppose $v$ is in a cluster set $S$ having at least two vertices or having a vertex with a giving edge, and let $C$ be the component of $G[S]$ containing $v$. By %Lemma~\ref{gen_Sgood}, either $\sum_{x \in C} (d(x) - 2) \ge k-1$, or a vertex of $C$ has a giving edge. 
Let $T$ be defined for $C$ in Step~4 of the discharging procedure.
By Lemma~\ref{gen_P}, $v$ is either adjacent to one of its supervisors, or is incident to a neutral edge, or has a non-empty cluster set. In each case, $v$ is adjacent to a vertex $w$ that is either the supervisor of $v$, or the other endpoint of a neutral edge, or one of the subordinates of $v$. 
By Lemmas \ref{gen_inter} and \ref{gen_neutral}, $w$ does not belong to $S$. By Lemmas~\ref{l_neutral} and~\ref{l_sup}, $v$ gives weight at most $\frac{M-2}{2}$ to $w$, and it does not give weight to $w$ if $w$ is the supervisor of $v$.
Let $v_1$, ..., $v_{d-1}$ be the neighbours of $v$ distinct from $w$. For the $v_i$'s that are not the neighbours of $v$ in $T$, $n_{Tvv_i}$ is not defined and we set $n_{Tvv_i} = 0$, and $v$ gives at most $M-2 = M - 2 - n_{Tvv_i}(4-\frac{3M}{2})$ to $v_i$.

Suppose that at least two of the $v_i$'s give weight to $v$ or are in a component of $T - v$ containing a vertex that has a giving edge. Those two vertices do not receive weight from $v$ by Lemma \ref{onlyone} and Step~4. As $v$ has initial weight equal to $d-M$ and gives at most $M-2$ to each of its neighbours distinct from $w$, the final weight of $v$ is at least $d-M - \frac{M-2}{2} - (d-3)(M-2) = (d-3)(3-M) + 4 - \frac{3M}{2} \ge 0$ since $M < \frac{8}{3}$.

Suppose one of the $v_i$'s, say $v_1$, gives weight to $v$ (by Step~3 or 4) or is into a component of $T - v$ containing a vertex that has a giving edge. By Lemma \ref{onlyone} and Step~4, $v$ does not give weight to $v_1$. If $e(v)=vv_1$, then $v$ receives $M-2$ from $v_1$. In that case, $n_{Tvv_i}$ is not defined and we set $n_{Tv_1v} = 0$. So by Step~3 or 4, $v$ receives at least $M - 2 - n_{Tv_1v}(4-\frac{3M}{2}) = \frac{4-M}{2} - (n_{Tv_1v}+1)(4-\frac{3M}{2})$ from $v_1$. Moreover $v$ gives at most $\frac{4-M}{2} - (n_{Tvv_i}+1)(4-\frac{3M}{2})$ to each of the other $v_i$'s. Observe that $\sum_{i > 1}(n_{Tvv_i}+1) \ge n_{Tv_1v}$ (since $v$ has one neighbour outside from the $v_i$'s). At last, $v$ gives at most $\frac{M-2}{2}$ to its last neighbour $w$. Overall, $v$ receives at least $\frac{4-M}{2} - (n_{Tv_1v}+1)(4-\frac{3M}{2})$ and gives at most $(d - 2)\frac{4-M}{2} - n_{Tv_1v} (4-\frac{3M}{2}) + \frac{M-2}{2}$, and initially it has a weight equal to $d-M$.  
Therefore, the final weight of $v$ is at least $d-M +\frac{4-M}{2} - (n_{Tv_1v}+1)(4-\frac{3M}{2}) - (d-2)\frac{4-M}{2} + n_{Tv_1v} (4-\frac{3M}{2}) - \frac{M-2}{2} = d-M - (d-3)\frac{4-M}{2} - (3-M) = (d-3)\frac{M-2}{2}$. Hence, if $d \ge 3$, then the final weight of $v$ is non-negative. 
If $d = 2$, then $v$ is adjacent to its supervisor by Lemma~\ref{gen_2sup} and $v$ gives nothing to its supervisor. Therefore the final weight of $v$ is at least $\frac{M-2}{2}$ more than what we computed above, and thus is non negative.

Suppose now that none of the $v_i$'s give weight to $v$ and that none of them are into a component of $T - v$ containing a vertex having a giving edge. That implies that no vertex of $C$ has a giving edge, $\sum_{x \in C} (d(x) - 2) \ge k-1$ by Lemma \ref{gen_Sgood}, and thus $\sum_{i} n_{Tvv_i} \ge k - d + 1$. For all $i$, $v$ gives at most $ M-2 - n_{Tvv_i}(4-\frac{3M}{2})$ to $v_i$. Moreover, $v$ gives at most $\frac{M-2}{2}$ to its last neighbour $w$. Overall, $v$ gives at most $(d-1)(M-2) - (k-d+1)(4-\frac{3M}{2}) + \frac{M-2}{2}$, and initially it has a weight equal to $d-M$. Therefore the final weight of $v$ is at least $d - M - (d-1)(M-2) + (k-d+1)(4-\frac{3M}{2}) - \frac{M-2}{2}$. As $M = \frac{8k}{3k+1}$, we have $k(4-\frac{3M}{2}) = \frac{M}{2}$, therefore the final weight of $v$ is at least $d - M - (d-1)(M - 2) + \frac{M}{2} -(d-1)(4-\frac{3M}{2})  - \frac{M-2}{2} = (d-3)\frac{M-2}{2}$. As in the previous paragraph, the final weight of $v$ is at least $0$ if $d\ge 3$. If $d = 2$, then $v$ is adjacent to its supervisor by Lemma~\ref{gen_2sup}, so $v$ does not give weight to its supervisor, and so its final weight is at least $(d-2)\frac{M-2}{2} = 0$. That ends the proof of Lemma~\ref{l_dech}.
\end{proof}

By Lemmas~\ref{gen_once}, \ref{l_unmarked}, and~\ref{l_dech}, every vertex has non-negative weight at the end of the procedure. As the sum of the weights of the vertices is negative, that leads to a contradiction. Therefore the counterexample $G$ to Theorem~\ref{t32_main} does not exist and that ends the proof of Theorem~\ref{t32_main}.

\section{Complexity results} \label{complex}
Let us first recall Theorem~\ref{kg}.
\begingroup
\def\thetheo{\ref{kg}}
\begin{theo}[recall]
Let $k\ge2$ and $g\ge3$ be fixed integers.
Either every planar graph with girth at least $g$ has an $({\cal I},{\cal P}_k)$-partition
or it is NP-complete to determine whether a planar graph with girth at least $g$ has an $({\cal I},{\cal P}_k)$-partition.
Either every planar graph with girth at least $g$ has an $({\cal I},{\cal O}_k)$-partition
or it is NP-complete to determine whether a planar graph with girth at least $g$ has an $({\cal I},{\cal O}_k)$-partition.
\addtocounter{theo}{-1}
\endgroup

\begin{proof}
The case $k=2$ has been considered in~\cite{EMOP11}, so we suppose that $k\ge3$.
The reduction is written in the context of $({\cal I},{\cal P}_k)$-partitions.
It will be straightforward to check that it also works in the context of $({\cal I},{\cal O}_k)$-partitions.
By Corollary~\ref{c_main}, we can assume that $g\le9$.

Let $g$ and $k$ be such that there exists a planar graph $G_{g,k}$ with girth at least $g$ that has no $({\cal I},{\cal P}_k)$-partition.
We construct a planar graph $H_{g,k}$ with girth at least $g$ such that $H_{g,d}$ has an $({\cal I},{\cal P}_k)$-partition
and $H_{g,k}$ contains a specific vertex $v$ that is in $I$ in every $({\cal I},{\cal P}_k)$-partition.
We can suppose that $G_{g,k}$ is minimal, that is, every proper subgraph of $G_{g,k}$ has an $({\cal I},{\cal P}_k)$-partition.
We consider the $({\cal I},{\cal P}_k)$-partitions of the graph $G^-_{g,k}$ obtained from $G_{g,k}$ by removing an arbitrary edge $xy$.
Since $G^-_{g,k}$ has an $({\cal I},{\cal P}_k)$-partition and $G_{g,k}$ does not, $x$ and $y$ must be have the same colour.
There are three possibilities.
\begin{itemize}
 \item $x$ and $y$ are coloured $I$ in every colouring. Then $H_{g,k}$ is $G^-_{g,k}$ and $v=x$.
 \item $x$ and $y$ are coloured $O$ in every colouring. Then $H_{g,k}$ is obtained from $G^-_{g,k}$ by adding the 2-vertex $v$ adjacent to $x$ and $y$. 
 \item There exists a colouring such that $x$ and $y$ are coloured $I$ and a colouring such that $x$ and $y$ are coloured $O$.
 Then $H_{g,k}$ is obtained as follows. Take three copies of $G^-_{g,k}$, identify their vertex $x$ to make one vertex $v$ and also identify their vertex $y$.
\end{itemize}
The graph $H'_{g,k}$ is obtained from $H_{g,k}$ by adding the 1-vertex $w$ adjacent to $v$.
So $H'_{g,k}$ has an $({\cal I},{\cal P}_k)$-partition and $w$ is coloured $O$ in every $({\cal I},{\cal P}_k)$-partition.
We reduce the NP-complete problem~\cite{DJPSY} \textsc{restricted planar $3$-sat}.
This variant of \textsc{sat} is such that:
\begin{itemize}
 \item every clause contains 2 or 3 literals,
 \item every variable occurs exactly twice positively and once negatively,
 \item the variable-clause incidence graph is planar.
\end{itemize}
A clause containing $t$ literals is called a $t$-clause.
Given an instance $K$ of \textsc{restricted planar $3$-sat}, we construct a graph $J$ corresponding to $K$.
The boolean value true corresponds to the colour $I$ and false corresponds to $O$.
For every variable of $K$, we put a copy of $P_3$, the variable gadget, in $J$.
Every vertex of the variable gadget corresponds to an occurrence of the variable: the extremities of $P_3$ 
correspond the two positive occurrences and the middle vertex corresponds to the negative occurrence.
For every clause of $K$, we put in $J$ one copy of the 2-clause gadget or the 3-clause gadget depicted in Figure~\ref{clause}.
The $t$-clause gadget consists in a path on $k+1$ vertices such that $t$ vertices correspond to the $t$ literals
of the clause and the other vertices are forced to be coloured $O$ by a copy of $H'_{g,k}$.
For every literal occurrence in $K$, we put in $J$ a copy the transmitter gadget depicted in Fig.~\ref{clause},
we identify the vertex $s$ of the transmitter with the vertex corresponding to the literal in the variable gadget,
and we identify the vertex $e$ of the transmitter with the vertex corresponding to the literal in the clause gadget.
The transmitter ensures that the girth condition on $J$ holds.
There exists a colouring such that both $s$ and $e$ are coloured $I$.
There exists a colouring such that both $s$ and $e$ are coloured $O$.
There exists no colouring such that $s$ is coloured $O$ and $e$ is coloured $I$.
\begin{figure}[htbp]
\centering
\includegraphics[width=13cm]{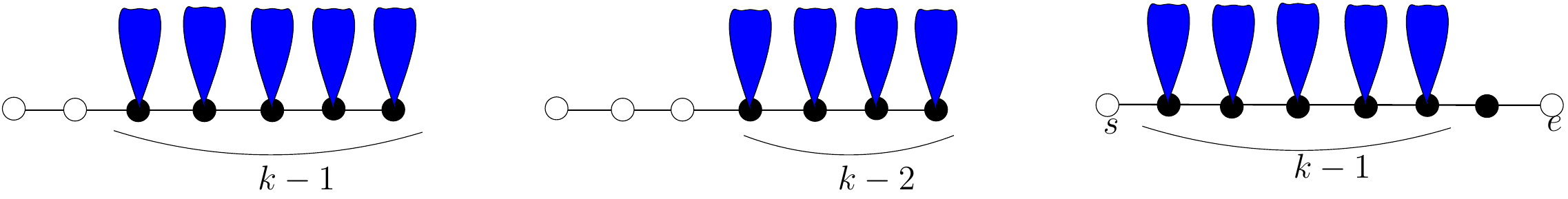}
\caption{The 2-clause gadget (left), the 3-clause gadget (middle), and the transmitter (right).}
\label{clause}
\end{figure}

\begin{figure}[htbp]
\centering
\includegraphics[width=8cm]{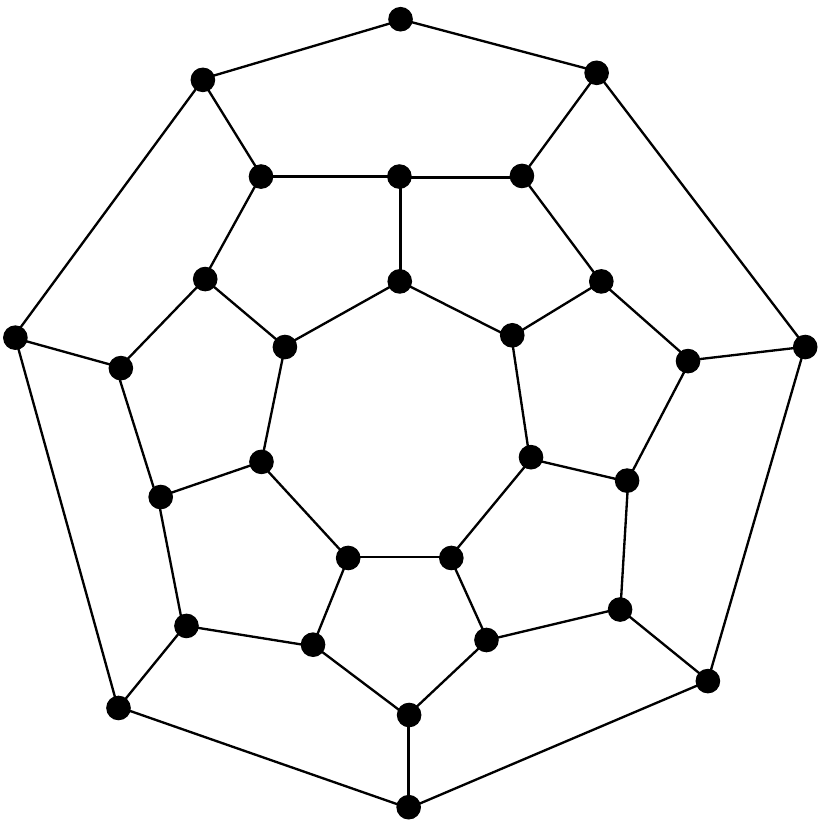}
\caption{A planar graph with girth $5$ and maximum degree $3$ that forces colour $I$ on the 2-vertex in every $({\cal I},{\cal P}_3)$-partition.}
\label{g5d3}
\end{figure}

\begin{figure}[htbp]
\centering
\includegraphics[width=10cm]{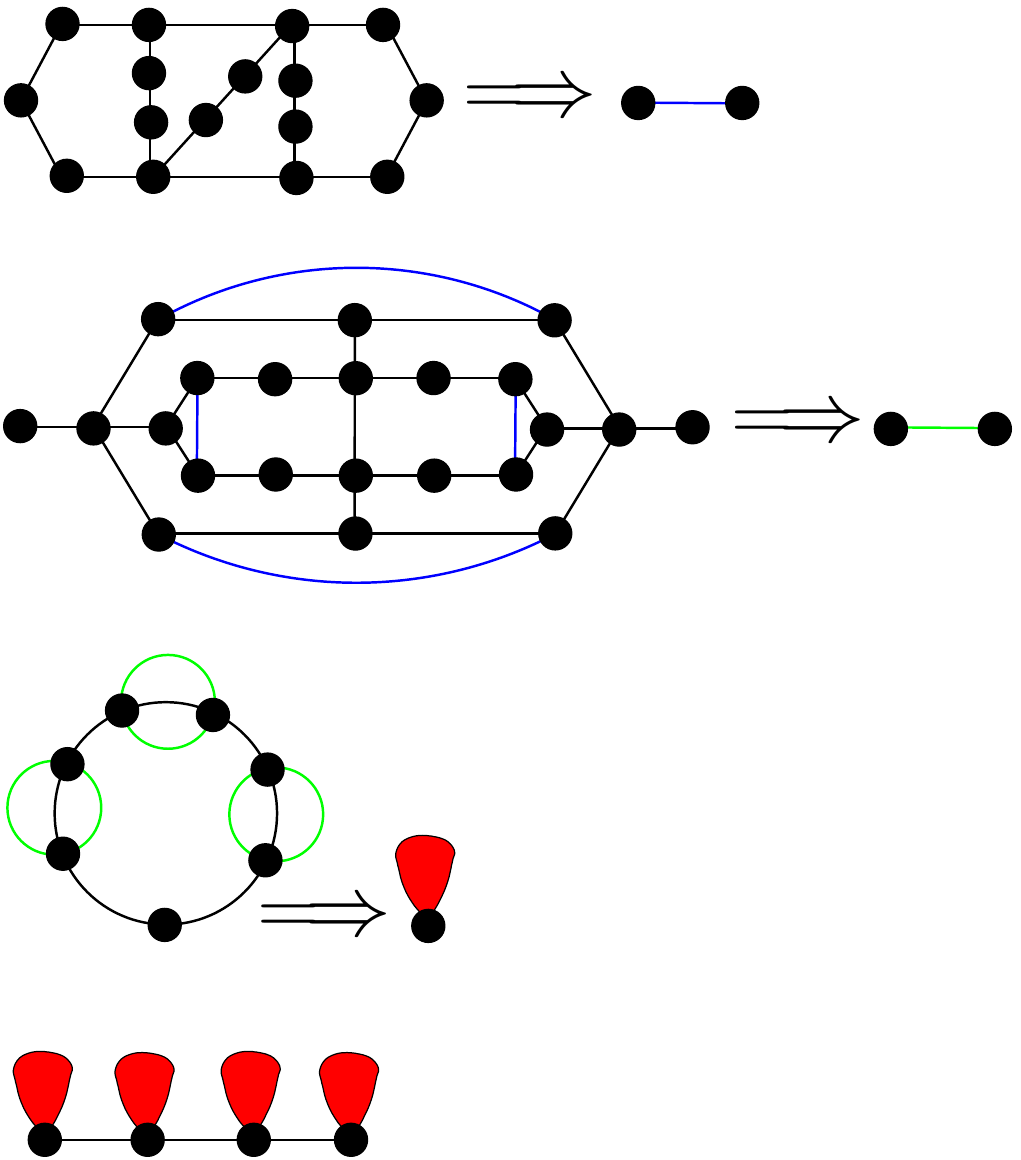}
\caption{A planar graph with girth $7$ and maximum degree $4$ that has no $({\cal I},{\cal P}_3)$-partition.}
\label{g7d4}
\end{figure}

Let us show that $K$ is satisfied if and only if $J$ has an $({\cal I},{\cal P}_k)$-partition.
Given an assignment satisfying $K$, we can colour accordingly the vertices corresponding to literals in the variable and the clause gadgets,
which gives an $({\cal I},{\cal P}_k)$-partition of $J$.
Given an $({\cal I},{\cal P}_k)$-partition of $J$, we assign to every variable the value true if and only if the middle vertex of the corresponding
variable gadget is coloured $O$, which gives an assignment satisfying $K$.
To finish the proof, observe that $J$ is planar with girth at least $g$.
\end{proof}

\begingroup
\def\thetheo{\ref{3gd}}
\begin{theo}[recall]
Let $d$ and $g\ge3$ be fixed integers.
Either every planar graph with girth at least $g$ and maximum degree at most $d$ has an $({\cal I},{\cal P}_3)$-partition
or it is NP-complete to determine whether a planar graph with girth at least $g$ and maximum degree at most $d$ has an $({\cal I},{\cal P}_3)$-partition.
\end{theo}
\addtocounter{theo}{-1}
\endgroup

\begin{proof}
The case $d\le2$ is easily settled, so we suppose that $d\ge3$.
The proof is a refinement of the case $k=3$ in Theorem~\ref{kg} such that the maximum degree is preserved.
Let $g$ and $d$ be such that there exists a planar graph $A_{g,d}$ with girth at least $g$ and maximum degre at most $d$ that has no $({\cal I},{\cal P}_3)$-partition.
We construct a planar graph $B_{g,d}$ with girth at least $g$ and maximum degree at most $d$ such that $B_{g,d}$ has an $({\cal I},{\cal P}_3)$-partition
and $B_{g,d}$ contains a specific vertex $v$ with degree at most 2 that is coloured $I$ in every $({\cal I},{\cal P}_3)$-partition.
Figure~\ref{g5d3} depicts $B_{5,3}$. This settles the case $g\le5$, so we can suppose that $g\ge6$.
Without loss generality, we suppose that $A_{g,d}$ is minimal, that is, every proper subgraph of $A_{g,d}$ has an $({\cal I},{\cal P}_3)$-partition.
Since $g\ge6$, $A_{g,d}$ is 2-degenerate. By minimality, $A_{g,d}$ does not contain a vertex of degree at most 1.
So, $A_{g,d}$ contains a 2-vertex $w$ and we call $v_1$ and $v_5$ its neighbours.
Let $A'_{g,d}$ be obtained from $A_{g,d}$ by removing $w$ and by adding three 2-vertices $v_2$, $v_3$, and $v_4$ which form the path
$v_1v_2v_3v_4v_5$. It is easy to check that in every $({\cal I},{\cal P}_3)$-partition of $A'_{g,d}$, $v_3$ is contained in a connected component
of size exactly two of the graph induced by the colour $O$. The graph $B_{g,d}$ is obtained from two copies of $A'_{g,d}$
by adding the 2-vertex $v$ adjacent to the vertex $v_3$ of both copies. Hence, $v$ must be coloured $I$.
The graph $A'_{g,d}$ is obtained from $A_{g,d}$ by adding the 1-vertex $w$ adjacent to $v$.
So, $A'_{g,d}$ forces the colour $O$ on its 1-vertex $w$ and plays the role of $H_{g,k}$ in the previous proof,
with the additional property that the maximum degree of $A'_{g,d}$ is at most $d$.
The reduction is then similar, using the gadgets for variables, clauses, and transmitters of the previous proof with $k=3$.
The obtained graph $J$ is indeed planar with girth at least $g$ and maximum degree at most $d$.
\end{proof}

\begingroup
\def\thetheo{\ref{c-3gd}}
\begin{cor}[recall]
Deciding whether a planar graph with girth at least $5$ and maximum degre at most $3$ is $({\cal I},{\cal P}_3)$-colourable is NP-complete.
Deciding whether a planar graph with girth at least $7$ and maximum degre at most $4$ is $({\cal I},{\cal P}_3)$-colourable is NP-complete.
\end{cor}
\addtocounter{theo}{-1}
\endgroup

\begin{proof}
The graph depicted in Figure~\ref{g5d3} is such that the 2-vertex at the top is coloured $I$ in every $({\cal I},{\cal P}_3)$-partition.
The graph depicted in Figure~\ref{g7d4} has no $({\cal I},{\cal P}_3)$-partition.
Using Theorem~\ref{3gd}, they allow to obtain Corollary~\ref{c-3gd}.

\end{proof}

\section*{Acknowledgements}

The authors thank Mickael Montassier for useful discussion and helpful comments.
The first author thanks the AlGCo team for their hospitality during his stay at LIRMM, Montpellier, France.

%%%%%%%%%%%%%%% BIB
%\bibliographystyle{alpha}
\bibliographystyle{plain}

%\bibliography{refsIO}

\begin{thebibliography}{10}

\bibitem{2003AlDiOpVe}
N.~Alon, G.~Ding, B.~Oporowski, and D.~Vertigan.
\newblock Partitioning into graphs with only small components.
\newblock {\em J. Combin. Theory Ser. B}, 87(2):231--243, 2003.

\bibitem{1977ApHa}
K.~Appel and W.~Haken.
\newblock Every planar map is four colorable. {I}. {D}ischarging.
\newblock {\em Illinois J. Math.}, 21(3):429--490, 1977.

\bibitem{1977ApHaKo}
K.~Appel, W.~Haken, and J.~Koch.
\newblock Every planar map is four colorable. {II}. {R}educibility.
\newblock {\em Illinois J. Math.}, 21(3):491--567, 1977.

\bibitem{2017AxUeWe}
M.~Axenovich, T.~Ueckerdt, and~P. Weiner.
\newblock Splitting planar graphs of girth 6 into two linear forests with short
  paths.
\newblock {\em J. Graph Theory}, 85(3):601--618, 2017.

\bibitem{2013BoKoYa}
O.~V.~Borodin, A.~V.~Kostochka, and M.~Yancey.
\newblock On 1-improper 2-coloring of sparse graphs.
\newblock {\em Discrete Math.}, 313(22):2638--2649, 2013.

\bibitem{2014BoKo}
O.~V.~Borodin and A.~V.~Kostochka.
\newblock Defective 2-colorings of sparse graphs.
\newblock {\em J. Combin. Theory Ser. B}, 104:72--80, 2014.

\bibitem{borodin2011list}
O.~V.~Borodin and A.~O.~Ivanova.
\newblock List strong linear 2-arboricity of sparse graphs.
\newblock {\em J. Graph Theory}, 67(2):83--90, 2011.

\bibitem{2017ChChJeSu}
H.~Choi, I.~Choi, J.~Jeong, and G.~Suh.
\newblock {$(1,k)$}-coloring of graphs with girth at least five on a surface.
\newblock {\em J. Graph Theory}, 84(4):521--535, 2017.

\bibitem{2015ChRa}
I.~Choi and A.~Raspaud.
\newblock Planar graphs with girth at least {$5$} are {$(3,5)$}-colorable.
\newblock {\em Discrete Math.}, 338(4):661--667, 2015.

\bibitem{DJPSY}
E.~Dahlhaus, D.~S. Johnson, C.~H. Papadimitriou, P.~D. Seymour, and
  M.~Yannakakis.
\newblock The complexity of multiterminal cuts.
\newblock {\em SIAM J. Comput.}, 23(4):864--894, 1994.

\bibitem{DvNo}
Z.~Dvořák and S.~Norin. 
\newblock Islands in minor-closed classes. I. Bounded treewidth and separators. 
\newblock arXiv preprint arXiv:1710.02727, 2017.

\bibitem{EMOP11}
L.~Esperet, M.~Montassier, P.~Ochem, and A.~Pinlou.
\newblock A complexity dichotomy for the coloring of sparse graphs.
\newblock {\em J. Graph Theory}, 73(1):85--102, 2013.

\bibitem{esperet2016islands}
L.~Esperet and P.~Ochem.
\newblock Islands in graphs on surfaces.
\newblock {\em SIAM J. Discrete Math.}, 30(1):206--219, 2016.

\bibitem{1991Goddard}
W.~Goddard.
\newblock Acyclic colorings of planar graphs.
\newblock {\em Discrete Math.}, 91(1):91--94, 1991.

\bibitem{1959Grotzsch}
H.~Gr{\"o}tzsch.
\newblock Zur {T}heorie der diskreten {G}ebilde. {VII}. {E}in {D}reifarbensatz
  f\"ur dreikreisfreie {N}etze auf der {K}ugel.
\newblock {\em Wiss. Z. Martin-Luther-Univ. Halle-Wittenberg. Math.-Nat.
  Reihe}, 8:109--120, 1958/1959.

\bibitem{2015MoOc}
M.~Montassier and P.~Ochem.
\newblock Near-colorings: non-colorable graphs and {NP}-completeness.
\newblock {\em Electron. J. Combin.}, 22(1):Paper 1.57, 13, 2015.

\bibitem{1990Poh}
K.~S. Poh.
\newblock On the linear vertex-arboricity of a planar graph.
\newblock {\em J. Graph Theory}, 14(1):73--75, 1990.

\bibitem{2000Skrekovski}
R.~\v{S}krekovski.
\newblock List improper colorings of planar graphs with prescribed girth.
\newblock {\em Discrete Math.}, 214(1-3):221--233, 2000.

\end{thebibliography}

%\begin{thebibliography}{1}
%\bibitem{borodin2011list}
%O.V. Borodin and A.O. Ivanova.
%\newblock List strong linear 2-arboricity of sparse graphs.
%\newblock {\em Journal of Graph Theory}, 67(2):83--90, 2011.
%
%
%\bibitem{DJPSY} E.~Dahlhaus, D. S.~Johnson, C. H. Papadimitriou, P. D. Seymour, and M. Yannakakis.
%The complexity of multiterminal cuts.
%\emph{SIAM J. Comput.} \textbf{23(4)} (1994), 864--894.
%
%
%\bibitem{esperet2016islands}
%L.~Esperet and P.~Ochem.
%\newblock Islands in graphs on surfaces.
%\newblock {\em SIAM Journal on Discrete Mathematics}, 30(1):206--219, 2016.
%
%
%\bibitem{EMOP11} L.~Esperet, M.~Montassier, P.~Ochem, and A.~Pinlou.
%\newblock A complexity dichotomy for the colouring of sparse graphs.
%\newblock \emph{J. Graph Theory}, 73(1):85--102, 2013.
%
%\end{thebibliography}
\end{document}